\documentclass[12pt,reqno]{amsart}

\usepackage{amssymb}
\usepackage{amsfonts}
\usepackage{amsthm}
\usepackage{amsmath}

\usepackage{txfonts}
\usepackage{upgreek}
\usepackage{graphicx}
\usepackage{color}
\usepackage[utf8]{inputenc}
\usepackage[T1]{fontenc}
\usepackage{empheq}
\usepackage{bbm}
\usepackage{mathrsfs}
\usepackage{cite}
\usepackage[bbgreekl]{mathbbol}
\usepackage{breqn}
\usepackage{hyperref}

\newlength\mytemplen
\newsavebox\mytempbox
\definecolor{myblue}{rgb}{.97,.97,1}

\makeatletter
\newcommand\mybluebox{%
    \@ifnextchar[
       {\@mybluebox}%
       {\@mybluebox[0pt]}}

\def\@mybluebox[#1]{%
    \@ifnextchar[
       {\@@mybluebox[#1]}%
       {\@@mybluebox[#1][0pt]}}

\def\@@mybluebox[#1][#2]#3{
    \sbox\mytempbox{#3}%
    \mytemplen\ht\mytempbox
    \advance\mytemplen #1\relax
    \ht\mytempbox\mytemplen
    \mytemplen\dp\mytempbox
    \advance\mytemplen #2\relax
    \dp\mytempbox\mytemplen
    \colorbox{myblue}{\hspace{1em}\usebox{\mytempbox}\hspace{1em}}}
\makeatother    

\newcommand{\q}{{ q}}

\newcommand{\ds}{\displaystyle}

\def\im{\mathop{\hbox{\rm Im}}\nolimits}
\newcommand{\EXP}{e}
\newcommand{\ii}{\mathsf{i}}
\newcommand{\II}{\mathrm{II}}
\newcommand{\I}{\mathrm{I}}

\def\be#1\ee{\begin{align}\begin{split}#1\end{split}\end{align}}
\def\beq#1\eeq{\begin{align}\begin{split}#1\end{split}\end{align}}

\newtheorem{thm}{Theorem}
\newtheorem{rem}{Remark}
\newtheorem{lem}{Lemma}
\newtheorem{cor}{Corollary}
\newtheorem{prop}{Proposition}
\newtheorem{definition}{Definition}

\newcommand{\smb}[1]{\mathbb{#1}}
\newcommand{\tilc}{\smb{c}} 
\newcommand{\tilx}{\smb{x}}
\newcommand{\tily}{\smb{y}}
\newcommand{\tilz}{\smb{z}}
\newcommand{\tilw}{\smb{w}}
\newcommand{\tilt}{\smb{t}}

\newcommand{\vt}{x}
\newcommand{\vT}{T}
\newcommand{\Gt}{\Gamma^+}
\newcommand{\gt}{\gamma^+}
\newcommand{\pt}{\phi^+}

\newcommand{\sumY}{\smb{S}}

\newcommand{\Painleve}{Painlev\'{e} }

\begin{document}

\title[Lens Generalisation of $\uptau$-functions for Elliptic Painlev\'{e} Equation]{Lens Generalisation of $\uptau$-functions for the
Elliptic Discrete Painlev\'{e} Equation}

\author{Andrew P.~Kels}
\address{(APK) Institute of Physics, University of Tokyo,
Komaba, Tokyo 153-8902, Japan}

\author{Masahito Yamazaki}
\address{(MY) Kavli Institute for the Physics and Mathematics of the Universe (WPI),
University of Tokyo, Kashiwa, Chiba 277-8583, Japan}

\begin{abstract}
We propose a new bilinear Hirota equation for $\uptau$-functions associated with the $E_8$ root lattice,
that provides a ``lens'' generalisation of the $\uptau$-functions for the elliptic discrete \Painleve equation.
Our equations are characterized by a positive integer $r$ in addition to the usual elliptic parameters,
and involve a mixture of continuous variables with additional discrete variables, the latter taking values on the $E_8$ root lattice. 
We construct explicit $W(E_7)$-invariant hypergeometric solutions of this bilinear Hirota equation, which are given in terms of 
elliptic hypergeometric sum/integrals.
\end{abstract}

\maketitle
\tableofcontents

\section{Introduction}

In the literature many variations of the differential and discrete (difference) \Painleve equations
have been found. These equations have been classified into rational, trigonometric and elliptic equations.
At the top level of the hierarchy is the elliptic discrete \Painleve equation with affine Weyl group symmetry of type $E_8^{(1)}$.
This equation has been obtained from geometric considerations \cite{Sakai}, and as a discrete system on the $E_8$ root lattice \cite{ORG}
(see \cite{MR1984002,MR2353465} for relation between the two approaches, and \cite{MR3609039} for a comprehensive survey).

In a recent work \cite{Noumi}, Noumi has given details of the construction of ORG $\uptau$-functions \cite{ORG} on the $E_8$ lattice for the elliptic discrete \Painleve equation. The goal of this paper is to present a generalization of Noumi's $\uptau$-function, along with solutions of this $\uptau$-function that are given in terms of elliptic hypergeometric sum/integrals.  The latter are generalisations of elliptic hypergeometric integrals, and depend on additional discrete parameters which enter as arguments of the lens elliptic gamma function.  Such functions first appeared in the study of supersymmetric gauge theories \cite{Benini:2011nc}, and in recent works several elliptic hypergeometric sum/integral formulas have been studied and proven from a mathematical point of view \cite{Kels:2015bda,rarified,Kels:2017toi}. These results motivate the construction of the corresponding lens $\uptau$-function of this paper, which involves two copies of the $E_8$ root lattice
and a positive integer parameter $r$, and the resulting equations depend on the usual continuous variables, as well as additional discrete variables on the $E_8$ root lattice. 
We propose a bilinear Hirota-type equation for the $\uptau$-function,
and construct explicit solutions of the bilinear equation in terms of an elliptic hypergeometric sum/integral
for a general value of the integer parameter $r$, which is fixed throughout the paper.
The hypergeometric $\uptau$-functions of this paper are expected to provide a solution for some (not yet known) generalisation of the elliptic discrete \Painleve equation.

For the case of $r=1$, the elliptic hypergeometric sum/integral used in this paper, reduces to the same elliptic hypergeometric integral which provides the hypergeometric solution of Noumi's $\uptau$-function \cite{Noumi}.  Then it might be expected that for $r=1$, the $\uptau$-function of this paper will also reduce to Noumi's $\uptau$-function.  Surprisingly this is not the case, since even for $r=1$ the Hirota equations (and solutions) will be seen to retain the dependence on the discrete variables on the $E_8$ root lattice.  The $\uptau$-function of \cite{Noumi} would then appear to correspond to a possible degenerate case, where there is no contribution of the discrete variables on the $E_8$ root lattice, in which case our $\uptau$-function will take values in subsets of $\mathbb{C}^8$, as is the case for \cite{Noumi}.  This is a rather interesting subtlety that arises here, and appears to be necessary for constructing solutions given in terms of the elliptic hypergeometric sum/integral, which will satisfy the bilinear relations and the invariance under the Weyl group $W(E_7)$.  

It is expected that the results of this paper will open up many possible future research directions. For example, it would be be interesting to
find an explicit Hamiltonian form of the discrete \Painleve equation associated to the $\uptau$-function of this paper, and to explore the various degenerations of the equations.  It would also be interesting to explore the geometric aspects of these equations
along the lines of Sakai's classification \cite{Sakai}.  In another direction, the lens elliptic gamma function, which is a central function for this paper, first appeared
in the study of four-dimensional $\mathcal{N}=1$ supersymmetric gauge theories on a circle times the lens space $S^3/\mathbb{Z}_r$ \cite{Benini:2011nc}. 
This connection suggests that there exists an interpretation of the results of this paper in terms of supersymmetric gauge theories
and associated integrable lattice models \cite{Bazhanov:2010kz,Spiridonov:2010em,Bazhanov:2011mz,Yamazaki:2012cp,Terashima:2012cx,Yamazaki:2013nra,Bazhanov:2013bh,Razamat:2013opa,Kels:2015bda,GahramanovKels,Yamazaki:2018xbx}.

The rest of this paper is organized as follows.
In Section \ref{sec.functions}, we provide definitions of the ``lens'' set of 
special functions, which generalise the special functions that appear in the theory of elliptic hypergeometric integrals.
In Section \ref{sec.integrals}, we define an  elliptic hypergeometric sum/integral for constructing the hypergeometric $\uptau$-function, and present the relevant identities that it satisfies.
In Section \ref{sec.tau}, we formulate the Hirota identities for the $\uptau$-function on the $E_8$ lattice,
which are then decomposed into the $W(E_7)$-orbits in Section \ref{sec.E7}.
In Section \ref{sec.main}, we state the main theorem of this paper (Theorem \ref{theorem.tau_n}),
which provides an explicit $W(E_7)$-invariant, lens elliptic hypergeometric solution of the $\uptau$-function.
The proof of the main theorem is provided in Section \ref{sec.proof}.
In the Appendices, we respectively present the derivation of the sum/integral transformation for $W(E_7)$ reflection, and provide a brief overview of the multiple Bernoulli polynomials, which are used for the definitions of the lens special functions.

{\it Acknowledgements:} The authors thank Yasuhiko Yamada for stimulation discussions, many useful suggestions and encouragement.
APK also thanks Yang Shi for helpful discussions.
APK is an overseas researcher under Postdoctoral Fellowship of Japan Society for the Promotion of Science (JSPS). 
MY was supported  in part by World Premier International Research Center Initiative (WPI), MEXT, Japan,
and by the JSPS Grant-in-Aid for Scientific Research No.\ 17KK0087.

\section{Lens Theta Functions and Lens Elliptic Gamma Function}\label{sec.functions}

In this section, the definitions of the special functions are given that play a central role in this paper.
Namely, these are the lens theta function,
the lens elliptic gamma function, and the lens triple gamma function. 

In this paper we use the two complex parameters $\sigma, \tau \in \mathbb{C}$, that satisfy
\begin{align}
\label{paramsdef}
\im(\sigma)\,,\; \im(\tau)>0\,.
\end{align}

Our equations will also depend on an additional integer parameter
\begin{align}
r=1,2,\ldots
\end{align}

In this paper, continuous and discrete variables are denoted by a pair $X=(x,\smb{x})$, where $x$, and $\smb{x}$, correspond to the continuous and discrete variables respectively. 

\subsection{Lens Theta Functions}

The two lens theta functions $\theta_\tau,\theta_\sigma$, are defined by \cite{Kels:2015bda,Kels:2017toi} 
\begin{align}
\label{lthtdef}
\begin{array}{lr}
\theta_\tau(z,\smb{z};\sigma, \tau):=\EXP^{\phi_{\tau}(z,\smb{z};\sigma, \tau)}\,\theta(\EXP^{-2\pi\ii z}\EXP^{2\pi\ii\tau \smb{z}}\,|\,\EXP^{2\pi\ii\tau r}), &\\[0.2cm]
\theta_\sigma(z,\smb{z}; \sigma, \tau):=\EXP^{\phi_{\sigma}(z,\smb{z};\sigma, \tau)}\,\theta(\EXP^{2\pi\ii z}\EXP^{2\pi\ii\sigma \smb{z}}\,|\,\EXP^{2\pi\ii\sigma r}), &
\multicolumn{1}{r}{\smash{\raisebox{.5\normalbaselineskip}{\hspace{0.3cm}$z\in\mathbb{C}$,\hspace{0.1cm} $\smb{z}\in\mathbb{Z}$,}}}
\end{array}
\end{align}
where for $\q=\EXP^{2\pi\ii\tau}$, $\theta(z\,|\,\q)$ is the regular theta function
\begin{align}
\label{regtheta}
\theta(z\,|\,\q)=\left(z;\q\right)_\infty\,(\q z^{-1};\q)_\infty,
\qquad
(z;\q)_\infty=\prod_{j=0}^\infty(1-z\q^j),
\end{align}
and the normalization factors are given by
\begin{align}
\begin{split}
&\phi_{\tau}(z,\smb{z};\sigma, \tau) 
\\
&\quad 
:=\frac{\pi\ii}{6r}\big(3(r+1-2\smb{z})(2z+1)-(r^2-1)(\sigma-\tau-1)-6\smb{z}(r-\smb{z})(\tau+1)\big) , \\
&\phi_{\sigma}(z,\smb{z}; \sigma, \tau)
\\
&\quad 
: =\frac{-\pi\ii}{6r}\big(3(r-1-2\smb{z})(2z-1)-(r^2-1)(\sigma-\tau-1)+6\smb{z}(r-\smb{z})(\sigma-1)\big) .
\end{split}
\label{lthtnorm}
\end{align}

For $r=1$, the lens theta functions \eqref{lthtdef} reduce to regular theta functions
\begin{align}
\label{thetar1}
\begin{split}
\ds\left.\theta_\tau(z,\smb{z};\sigma, \tau)\,\right|_{r=1}=\theta(\EXP^{2\pi\ii z}\,|\,\EXP^{2\pi\ii \tau}), \\
\ds\left.\theta_\sigma(z,\smb{z};\sigma, \tau)\,\right|_{r=1}=\theta(\EXP^{2\pi\ii z}\,|\,\EXP^{2\pi\ii \sigma}).
\end{split}
\end{align}

Note that the theta functions $\theta_\tau$ and $\theta_\sigma$ defined in \eqref{lthtdef}, each have non-trivial dependence on both of the parameters $\sigma$, and $\tau$, through the normalization functions \eqref{lthtnorm}.

For brevity, the lens theta functions \eqref{lthtdef} will typically be denoted by 
\begin{align}
\theta_\tau(z,\smb{z}):=\theta_\tau(z,\smb{z};\sigma, \tau)\,,\qquad \theta_\sigma(z,\smb{z}):=\theta_\sigma(z,\smb{z}; \sigma, \tau),
\end{align}
with implicit dependence on the two parameters $\sigma$ and $\tau$.

Furthermore, a shorthand notation will be used throughout this paper, where $\pm$ in the argument of a function denotes that respective factors involving $+$ and $-$ should be taken as a product, {\it e.g.}
\begin{align}\label{compnot}
\theta_\sigma(\vt_j\pm \vt_k,\smb{x}_j\pm \smb{x}_k)=\theta_\sigma(\vt_j+ \vt_k,\smb{x}_j+\smb{x}_k) \, \theta_\sigma(\vt_j- \vt_k,\smb{x}_j- \smb{x}_k) .
\end{align}

\begin{prop}\label{prop.theta}

The lens theta functions satisfy 
(here $\theta_{\tau,\sigma}$ indicates that an identity holds for either $\theta_\tau$ or $\theta_\sigma$):
\medskip

(1) (periodicity)
For $k\in\mathbb{Z}$,
\begin{align}
\begin{split}
\theta_{\tau,\sigma}(z+2kr,\smb{z})=\theta_{\tau,\sigma}(z,\smb{z}),\qquad\theta_{\tau,\sigma}(z,\smb{z}+kr)=\theta_{\tau,\sigma}(z,\smb{z}).
\end{split}
\label{periodic_theta}
\end{align}

(2) (inversion)
\begin{align}
\label{thtinv}
\begin{split}
\theta_{\tau,\sigma}(-z,-\smb{z})&=-\theta_{\tau,\sigma}(z,\smb{z})\,\EXP^{-\frac{2\pi\ii}{r}(z-\smb{z})}.
\end{split}
\end{align}

(3) (recurrence relation)
For $n\in\mathbb{Z}$,
\begin{align}
\label{thtshft}
\begin{split}
\theta_\tau(z+n\tau,\smb{z}+n)&=\theta_\tau(z,\smb{z})\,\EXP^{-\frac{n\pi\ii}{r}(2z+(n-1)\tau+r-2\smb{z}-n+1)},\\
\theta_\sigma(z+n\sigma,\smb{z}-n)&=\theta_\sigma(z,\smb{z})\,\EXP^{-\frac{n\pi\ii}{r}(2z+(n-1)\sigma+r-2\smb{z}+n-1)}.
\end{split}
\end{align}

(4) (quasi-periodicity)
For $n\in\mathbb{Z}$,
\begin{align}
\label{thtshft2}
\begin{split}
\theta_\tau(z+rn\tau,\smb{z})&=\theta_\tau(z,\smb{z})\,\EXP^{-n\pi\ii(2z+\tau(rn-1)+1)},\\
\theta_\sigma(z+rn\sigma,\smb{z})&=\theta_\sigma(z,\smb{z})\,\EXP^{-n\pi\ii(2z+\sigma(rn-1)+1)}.
\end{split}
\end{align}

(5) (three-term relation)
For $\vt_i,\vt_j,\vt_k,z\in\mathbb{C}$, and $\smb{x}_i,\smb{x}_j,\smb{x}_k, \smb{z}\in \mathbb{Z}$, or  $\smb{x}_i,\smb{x}_j,\smb{x}_k, \smb{z}\in \mathbb{Z}+\frac{1}{2}$,
\begin{align}
\label{tht13term}
\begin{split}
&\EXP^{\frac{2\pi\ii}{r}(\vt_k-\smb{x}_k)}\theta_{\tau,\sigma}(\vt_j\pm \vt_k,\smb{x}_j\pm \smb{x}_k)\,\theta_{\tau,\sigma}(\vt_i\pm z,\smb{x}_i\pm \smb{z}) \\
&+\EXP^{\frac{2\pi\ii}{r}(\vt_i-\smb{x}_i)}\theta_{\tau,\sigma}(\vt_k\pm \vt_i,\smb{x}_k\pm \smb{x}_i)\,\theta_{\tau,\sigma}(\vt_j\pm z,\smb{x}_j\pm \smb{z}) \\
&+\EXP^{\frac{2\pi\ii}{r}(\vt_j-\smb{x}_j)}\theta_{\tau,\sigma}(\vt_i\pm \vt_j,\smb{x}_i\pm \smb{x}_j)\,\theta_{\tau,\sigma}(\vt_k\pm z,\smb{x}_k\pm \smb{z})=0.
\end{split}
\end{align}

\end{prop}

\begin{proof}
These identities simply follow from the definitions \eqref{lthtdef}, and
similar identities that hold for the regular theta function $\theta(z|\q)$, defined in \eqref{regtheta}.
\end{proof}

\subsection{Lens Elliptic Gamma Function}

The lens elliptic gamma function \cite{Benini:2011nc,Razamat:2013opa,Kels:2015bda,GahramanovKels} is defined here by
\begin{align} 
\label{legf2}
\Gamma(z,\smb{z};\sigma,\tau):=\EXP^{\phi_e(z, \smb{z};\sigma,\tau)}\gamma_\sigma(z,\smb{z};\sigma,\tau)\gamma_\tau(z,\smb{z};\sigma,\tau),\quad z\in\mathbb{C}, \; \smb{z}\in \mathbb{Z},
\end{align}
where $\gamma_\sigma$ and $\gamma_\tau$, are the following infinite products
\begin{align}
\label{littlegammap}
\gamma_\sigma(z,\smb{z};\sigma,\tau)&:=\prod_{j,k=0}^\infty\frac{1-\EXP^{-2\pi\ii z}\EXP^{-2\pi\ii \sigma \smb{z}}\EXP^{2\pi\ii(\sigma+\tau)(j+1)}\EXP^{2\pi\ii \sigma r(k+1)}}{1-\EXP^{2\pi\ii z}\EXP^{2\pi\ii\sigma \smb{z}}\EXP^{2\pi\ii(\sigma+\tau)j}\EXP^{2\pi\ii\sigma rk}}, \\
\label{littlegammaq}
\gamma_\tau(z,\smb{z};\sigma,\tau)&:=\prod_{j,k=0}^\infty\frac{1-\EXP^{-2\pi\ii z}\EXP^{-2\pi\ii\tau(r-\smb{z})}\EXP^{2\pi\ii(\sigma+\tau)(j+1)}\EXP^{2\pi\ii\tau r(k+1)}}{1-\EXP^{2\pi\ii z}\EXP^{2\pi\ii\tau(r-\smb{z})}\EXP^{2\pi\ii(\sigma+\tau)j}\EXP^{2\pi\ii\tau rk}},
\end{align}
and the normalisation function $\phi_e(z,\smb{z};\sigma,\tau)$  is given by \cite{Kels:2015bda,GahramanovKels,Kels:2017toi}
\begin{align}
\label{ellnormx}
\phi_e(z, \smb{z};\sigma,\tau)&:=\pi \ii \frac{\smb{z} (\smb{z}-r) (6 z-3\sigma-3\tau+(1-\sigma+\tau)  (r-2 \smb{z}))}{6 r}.
\end{align}

Note that the functions \eqref{littlegammap}, \eqref{littlegammaq}, are symmetric with respect to the following shifts
\begin{align}
\label{litgamrels}
\begin{split}
\gamma_\sigma(z+k\sigma,\smb{z}-k;\sigma,\tau)&=\gamma_\sigma(z,\smb{z};\sigma,\tau), \\
 \gamma_\tau(z+k\tau,\smb{z}+k;\sigma,\tau)&=\gamma_\tau(z,\smb{z};\sigma,\tau),
 \end{split}
\end{align}
for $k\in\mathbb{Z}$.

The normalisation function \eqref{ellnormx}, has a useful factorisation in terms of the multiple Bernoulli polynomial $B_{3,3}$ \eqref{bernoulli}, as \cite{GahramanovKels,Kels:2017toi}
\begin{align}
\label{ellnorm}
\begin{split}
\ds\phi_e(z,\smb{z};\sigma,\tau)&=\ds 2\pi\ii\left(R_2\left(z,0; \sigma-\frac{1}{2},\tau+\frac{1}{2}\right)-R_2\left(z,\smb{z};\sigma-\frac{1}{2},\tau+\frac{1}{2}\right)\right)\\[0.1cm]
&=\ds 2\pi\ii\left(R_2(z,0;\sigma,\tau)+R_2\left(0,\smb{z}; \frac{1}{2},-\frac{1}{2}\right)-R_2(z,\smb{z};\sigma,\tau)\right),
\end{split}
\end{align}
where 
\begin{align}
\label{r2def}
\begin{split}
\ds R_2(z,\smb{z};\sigma,\tau):&=\ds R(z+\smb{z}\sigma;r\sigma,\sigma+\tau)+R(z+(r-\smb{z})\tau;r\tau,\sigma+\tau),
\end{split}
\end{align}
and
\begin{align}
R(z;\sigma,\tau):&=\frac{B_{3,3}(z;\sigma,\tau,-1)+B_{3,3}(z-1;\sigma,\tau,-1)}{12}.
\end{align}

For $r=1$, the lens elliptic gamma function \eqref{legf2} reduces to the regular elliptic gamma function \cite{Ruijsenaars:1997:FOA}, which is denoted here by $\Gamma_{1}(z;\sigma,\tau)$,
\begin{align}
\label{egf2}
\ds\left.\Gamma(z,\smb{z};\sigma,\tau)\,\right|_{r=1}=\Gamma_{1}(z;\sigma,\tau)=
\ds\prod_{j,k=0}^\infty\frac{1-\EXP^{2\pi\ii (-z+(j+1) \sigma+(k+1)\tau)}}{1-\EXP^{2\pi\ii (z+j\sigma+k \tau)}}.
\end{align}
In terms of the regular elliptic gamma function \eqref{egf2}, the functions \eqref{littlegammap}, and \eqref{littlegammaq}, are simply
\begin{align}
\label{legfprod}
\begin{split}
\gamma_\sigma(z,\smb{z};\sigma,\tau)&=\Gamma_{1}(z+\sigma \smb{z};r\sigma,\sigma+\tau), \\
\gamma_\tau(z,\smb{z};\sigma,\tau)&=\Gamma_{1}(z+\tau(r-\smb{z});r\tau,\sigma+\tau).
\end{split}
\end{align}

Similarly to the lens theta functions, the lens elliptic gamma function \eqref{legf2} will typically be denoted as 
\begin{align}
\Gamma(z,\smb{z}):=\Gamma(z,\smb{z};\sigma, \tau),
\end{align}
with implicit dependence on the two parameters $\sigma$, and $\tau$.

\begin{prop}

The lens elliptic gamma function \eqref{legf2} satisfies
\medskip

(1) (periodicity)
For $k\in\mathbb{Z}$,
\begin{align}
\Gamma(z+2kr,\smb{z})=\Gamma(z,\smb{z}),\quad\Gamma(z,\smb{z}+kr)=\Gamma(z,\smb{z}).
\label{periodic_gamma}
\end{align}

(2) (inversion)
\begin{align}
&\ds\Gamma((\sigma+\tau)-z,-\smb{z})\,\Gamma(z,\smb{z})=1.\label{Gamma_invert} 
\end{align}

(3) (recurrence relation)
For $n=0,1\ldots$,
\begin{align}
&\ds\Gamma(z+n\sigma,\smb{z}-n)=\Gamma(z,\smb{z})\prod_{j=0}^{n-1}\theta_\tau(z+j\sigma,\smb{z}-j),  
\label{Gamma_ratio_1}\\
&\ds\Gamma(z+n\tau,\smb{z}+n)=\Gamma(z,\smb{z})\prod_{j=0}^{n-1}\theta_\sigma(z+j\tau,\smb{z}+j). 
 \label{Gamma_ratio_2}
\end{align}

\end{prop}

\begin{proof}
These identities can be verified by direct computation.  A proof of the $r$-periodicity in \eqref{periodic_gamma}
previously appeared in Appendix C of \cite{GahramanovKels}. For \eqref{Gamma_ratio_1}, \eqref{Gamma_ratio_2}, 
the normalisation of the lens theta functions \eqref{lthtnorm} are in fact chosen to satisfy
\begin{align}
\begin{split}
\phi_e(z+\sigma, \smb{z}-1;\sigma,\tau) -\phi_e(z, \smb{z};\sigma,\tau) =\phi_{\tau}(z, \smb{z};\sigma,\tau) , \\
\phi_e(z+\tau, \smb{z}+1;\sigma,\tau) -\phi_e(z, \smb{z};\sigma,\tau) =\phi_{\sigma}(z, \smb{z};\sigma,\tau).
\end{split}
\end{align}
Due to the relations \eqref{litgamrels}, only the factor of $\gamma_\tau$ on the left hand side of \eqref{Gamma_ratio_1}, contributes to the infinite product part of the theta function $\theta_\tau$, 
while only the factor of $\gamma_\sigma$ on the left hand side of \eqref{Gamma_ratio_2}, contributes to the infinite product part of the theta function $\theta_\sigma$.
\end{proof}

\begin{rem}
Although $\gamma_{\sigma,\tau}(z+k,\smb{z};\sigma,\tau)=\gamma_{\sigma,\tau}(z,\smb{z};\sigma,\tau)$ for any $k\in\mathbb{Z}$, 
the $2r$-periodicity of $\Gamma(z,\smb{z}; \sigma, \tau)$ in \eqref{periodic_gamma} comes from the normalisation factor  in \eqref{ellnorm}.
\end{rem}

\subsection{Lens Triple Gamma Functions}

Here we consider two parameters $\im(\omega), \im(\mu)>0$, in addition to the parameters \eqref{paramsdef}.

The lens triple gamma functions $\Gt_\sigma$ and $\Gt_{\sigma \tau}$, are defined here by
\begin{align}
\label{Gam3def}
\begin{array}{lr}
\Gt_\sigma(z,\smb{z};\sigma,\tau,\omega):=\EXP^{\pt(z,\smb{z};\sigma,\tau,\omega)}\gt_\sigma(z,\smb{z};\sigma,\tau,\omega), & \\[0.2cm]
\Gt_{\sigma\tau}(z,\smb{z};\sigma,\tau,\omega,\mu):=\Gt_\sigma(z,\smb{z};\sigma,\tau,\omega)\gt_\tau(z,\smb{z};\sigma,\tau,\mu), &
\multicolumn{1}{r}{\smash{\raisebox{.5\normalbaselineskip}{\hspace{0.3cm}$z\in\mathbb{C}$,\hspace{0.1cm} $\smb{z}\in\mathbb{Z}$,}}}
\end{array}
\end{align}
where
\begin{align}
\label{lgam3def}
\begin{split}
\gt_\sigma(z,\smb{z};\sigma,\tau,\omega)&:=g_\sigma(z,\smb{z};\sigma,\tau,\omega)\, g_\sigma(\sigma+\tau+\omega-z,r+1-\smb{z};\sigma,\tau,\omega), \\
\gt_\tau(z,\smb{z};\sigma,\tau,\omega)&:=g_\tau(z,\smb{z};\sigma,\tau,\omega)\, g_\tau(\sigma+\tau+\omega-z,r-\smb{z};\sigma,\tau,\omega),
\end{split}
\end{align}
and
\begin{align}
\begin{split}
g_\sigma(z,\smb{z};\sigma,\tau,\omega):=&\prod_{k_1,k_2,k_3=0}^\infty\left(1-\EXP^{2\pi\ii x}\EXP^{2\pi\ii\sigma \tilz}\EXP^{2\pi\ii\sigma rk_1}\EXP^{2\pi\ii(\sigma+\tau)k_2}\EXP^{2\pi\ii(\sigma+\omega)k_3}\right), \\
g_\tau(z,\smb{z};\sigma,\tau,\omega):=&\prod_{k_1,k_2,k_3=0}^\infty\left(1-\EXP^{2\pi\ii x}\EXP^{2\pi\ii\sigma(r-\tilz)}\EXP^{2\pi\ii\tau rk_1}\EXP^{2\pi\ii(\sigma+\tau)k_2}\EXP^{2\pi\ii\omega k_3}\right).
\end{split}
\end{align}

The normalisation function  $\pt(z,\smb{z};\sigma,\tau,\omega)$ is defined by ({\it c.f.} the expression \eqref{ellnorm} for $\phi_e$ in terms of $B_{3,3}$)
\begin{align}
\label{Gam3norm}
\begin{split}
&\ds\pt(z,\smb{z};\sigma,\tau,\omega)\\
&\quad:=\ds 2\pi\ii\left(T_2\left(z,0; \sigma-\frac{1}{2},\tau+\frac{1}{2},\omega\right)-S_2\left(z,\smb{z};\sigma-\frac{1}{2},\tau+\frac{1}{2},\omega\right)\right),
\end{split}
\end{align}
where
\begin{align}
\begin{split}
&\ds S_2(z,\smb{z};\sigma,\tau,\omega)\\
&\quad :=\ds S(z+\smb{z}\sigma;r\sigma,\sigma+\tau,\omega+\sigma)+S(z+(r-\smb{z})\tau;r\tau,\sigma+\tau,\omega-\tau), \\
&T_2(z,\smb{z};\sigma,\tau,\omega)\\
&\quad:=\ds S(z+\smb{z}\sigma;r\sigma,\sigma+\tau,\omega)+S(z+(r-\smb{z})\tau;r\tau,\sigma+\tau,\omega), 
\end{split}
\end{align}
and
\begin{align}
S(z;\sigma,\tau,\omega):=\frac{B_{4,4}(z;\sigma,\tau,-1,\omega)+B_{4,4}(z-1;\sigma,\tau,-1,\omega)}{48}.
\end{align}
In the last equation, $B_{4,4}(z;\omega_1,\omega_2,\omega_3,\omega_4)$ is the multiple Bernoulli polynomial \eqref{bernoulli4}, defined in Appendix \ref{app.Bernoulli}.

\begin{prop}\label{lgam3prop}
The functions $\gamma_\sigma$, $\gamma_\tau$, defined in \eqref{lgam3def}, satisfy 

\medskip

(1) (shift symmetry)
\begin{align} \label{shift_lgam3}
\begin{split}
\gt_\sigma(z+\sigma,\smb{z}-1;\sigma,\tau,\omega)=\gt_\sigma(z,\smb{z};\sigma,\tau,\omega) , \\
\gt_\tau(z+\tau,\smb{z}+1;\sigma,\tau,\omega)=\gt_\tau(z,\smb{z};\sigma,\tau,\omega) .
\end{split}
\end{align}

(2) (inversion)
\begin{align} \label{invert_lgam3}
\begin{split}
\gt_\sigma(\sigma+\tau+\omega-z,r+1-\smb{z};\sigma,\tau,\omega)&=\gt_\sigma(z,\smb{z};\sigma,\tau,\omega), \\
\gt_\tau(\sigma+\tau+\omega-z,r-\smb{z};\sigma,\tau,\omega)&=\gt_\tau(z,\smb{z};\sigma,\tau,\omega).
\end{split}
\end{align}

(3) (recurrence relation)
\begin{align} \label{recurr_lgam3}
\begin{split}
\gt_\sigma(\omega+z,1+\smb{z};\sigma,\tau,\omega)& =\gamma_\sigma(z,\smb{z};\sigma,\tau)\, \gt_\sigma(z,\smb{z};\sigma,\tau,\omega), \\
\gt_\tau(\omega+z,\smb{z};\sigma,\tau,\omega)& =\gamma_\tau(z,\smb{z};\sigma,\tau)\, \gt_\tau(z,\smb{z};\sigma,\tau,\omega) .
\end{split}
\end{align}
\end{prop}

\begin{proof}
These relations essentially follow from the definitions given in \eqref{lgam3def}.
\end{proof}

\begin{cor}
The lens triple gamma functions \eqref{Gam3def} satisfy 
\medskip

(1) (inversion)
\begin{align}
\label{invert_Gamma3}
\Gt_\sigma(\sigma+\tau+\omega-z,r+1-\smb{z};\sigma,\tau,\omega)=\Gt_\sigma(z,\smb{z};\sigma,\tau,\omega) ,
\end{align}

(2) (recurrence relation)
\begin{align}
\label{shift_Gamma3}
\begin{split}
\Gt_\sigma(\omega+z,1+\smb{z};\sigma,\tau,\omega)=\EXP^{\phi_e(z,\smb{z};\sigma,\tau)}\gamma_\sigma(z,\smb{z};\sigma,\tau) \,\Gt_\sigma(z,\smb{z};\sigma,\tau,\omega), \\
\Gt_{\sigma\tau}(\sigma+\tau-z,r-\smb{z};\sigma,\tau,\omega,\mu)=\Gamma(z,\smb{z};\sigma,\tau)\,\Gt_{\sigma\tau}(z,\smb{z};\sigma,\tau,\omega,\mu),
\end{split}
\end{align}
where $\phi_e$, and $\gamma_\sigma$ are defined respectively in \eqref{ellnorm}, and \eqref{littlegammap}, and $\Gamma(z,\smb{z};\sigma,\tau)$ is the lens elliptic gamma function \eqref{legf2}.
\end{cor}

\begin{proof}
The relations \eqref{invert_Gamma3}, \eqref{shift_Gamma3} follow from the relations given in Proposition \ref{lgam3prop}, and also the following relations satisfied by the normalisation function \eqref{Gam3norm}
\begin{align}
\begin{split}
&\pt(\sigma+\tau+\omega-z,r+1-\smb{z};\sigma,\tau,\omega)=\pt(z,\smb{z};\sigma,\tau,\omega), \\
&\pt(\omega+z,\smb{z}+1;\sigma,\tau,\omega)=\pt(z,\smb{z};\sigma,\tau,\omega)+\phi_e(z,\smb{z},\sigma,\tau),
\end{split}
\end{align}
where $\phi_e(z,\smb{z},\sigma,\tau)$ is the normalisation function for the lens elliptic gamma function given in \eqref{ellnorm}.
\end{proof}

\begin{rem}
Note that unlike the lens theta and elliptic gamma functions, the lens triple gamma functions \eqref{Gam3def} are not $r$-periodic in $\smb{z}$, and even 
for $r=1$ there remains a dependence on the integer variable $\smb{z}$.  This is the reason why the hypergeometric $\uptau$-function constructed in Section \ref{sec.main}, retains the dependence on the discrete variables even for $r=1$.
\end{rem}

\section{\texorpdfstring{Elliptic Hypergeometric Sum/Integral and $W(E_7)$ transformation}{Elliptic Hypergeometric Sum/Integral and W(E7) transformation}}\label{sec.integrals}

\subsection{Elliptic Hypergeometric Sum/Integral}

A central role in this paper is played by the following sum/integral, defined in terms of the lens elliptic gamma function \eqref{legf2}, by
\begin{align}
\label{ebsi2}
I(\vt,\smb{\vt}; \sigma, \tau)=\frac{\lambda(\sigma, \tau)}{2}\sum^{r-1}_{\smb{z}=0}\int_{[0,1]}dz\,\frac{\prod_{j=0}^{7}\Gamma(\vt_j\pm z,\smb{\vt}_j\pm \tilde{\smb{z}})}{\Gamma(\pm 2z,\pm 2\tilde{\smb{z}})},
\end{align}
where $\vt=(\vt_{0}, \dots, \vt_{7})\in \mathbb{C}^8$, $\im(x_i)>0$, and $\smb{\vt}=(\smb{\vt}_{0}, \dots, \smb{\vt}_{7})\in \mathbb{Z}^8 \cup (\mathbb{Z}+\frac{1}{2})^8$.
Notice that in contrast to the previous section, here we allow $\smb{x}$ to have either integer, or half-integer components.
The discrete summation variable is chosen so that the second argument of each factor of the lens elliptic gamma functions appearing in \eqref{ebsi2} is an integer, and is defined by
\begin{align}
\label{ztdef}
\begin{split}
\tilde{\smb{z}}:=&\,\smb{z}+(r+((r+1)\bmod 2))(\smb{\vt}_1 \bmod 1) \\
=&
\begin{cases}
\smb{z} & (\smb{x}\in \mathbb{Z}^8) ,\\
\smb{z}+ \frac{r+1}{2}& \left(\smb{x}\in (\mathbb{Z}+\frac{1}{2})^8, \,  r \textrm{ even } \right),\\
\smb{z}+ \frac{r}{2}& \left(\smb{x}\in (\mathbb{Z}+\frac{1}{2})^8, \,  r \textrm{ odd } \right).
\end{cases} 
\end{split}
\end{align}
The prefactor $\lambda(\sigma, \tau)$ in \eqref{ebsi2} is given by
\begin{align}
\lambda(\sigma, \tau) = (e^{2\pi \ii r \sigma};e^{2\pi \ii r \sigma})_{\infty} (e^{2\pi \ii r \tau};e^{2\pi \ii r \tau})_{\infty}.
\label{lambda_def}
\end{align}
The condition $\im(x_i)>0$ may be relaxed, by deforming the contour connecting the points $z=0$, and $z=1$, such that the respective poles of the integrand of \eqref{ebsi2} do not cross over the contour \cite{Kels:2017toi}.

For $\smb{\vt}\in \mathbb{Z}^8$, the elliptic hypergeometric sum/integral \eqref{ebsi2} previously appeared as part of a key identity (star-star relation) for the integrability of multi-spin lattice models \cite{Yamazaki:2013nra,Kels:2017toi}, and is a 2-parameter extension of the left hand side of the elliptic beta sum/integral formula that was proven in \cite{Kels:2015bda}.  It has also previously been studied with respect to $A_1\leftrightarrow A_n$, and $BC_1\leftrightarrow BC_n$ transformations proven by the authors \cite{Kels:2017toi} (where the $BC_1\leftrightarrow BC_0$ transformation was previously proven by Spiridonov \cite{rarified}, and the $r=1$ cases of the transformations were previously proven by Rains \cite{RainsT}).

\subsection{Contiguity Relation}
Define the shift operator $T_{\tau,k}$ ($k\in\{0,\ldots,7\}$), that acts on the continuous variables $\vt_k\in\mathbb{C}$, and discrete variables $\smb{\vt}_k\in\mathbb{Z},(\mathbb{Z}+\frac{1}{2})$, as
\begin{align}
T_{\tau,k}f(\vt_0, \dots,\vt_7, \smb{x}_0,\dots,\smb{x}_7)=f(\vt_0,\dots, \vt_k+\tau, \vt_7,\smb{\vt}_0,\dots,\smb{\vt}_k+1,\dots,\smb{x}_7).
\label{T_shift}
\end{align}

\begin{prop}

The elliptic hypergeometric sum/integral \eqref{ebsi2} satisfies the three-term relation
\begin{align}
\label{contrel}
\begin{split}
&\left(\EXP^{\frac{2\pi\ii}{r}(\vt_k-\smb{x}_k)}\theta_\sigma(\vt_j\pm \vt_k,\smb{\vt}_j\pm \smb{\vt}_k)\,T_{\tau,i}+\EXP^{\frac{2\pi\ii}{r}(\vt_i-\smb{x}_i)}\theta_\sigma(\vt_k\pm \vt_i,\smb{\vt}_k\pm \smb{\vt}_i)\,T_{\tau,j}\right. \\
&+\left.\EXP^{\frac{2\pi\ii}{r}(\vt_j-\smb{x}_j)}\theta_\sigma(\vt_i\pm \vt_j,\smb{\vt}_i\pm \smb{\vt}_j)\,T_{\tau,k}\right)\, I(\vt,\smb{\vt})=0,
\end{split}
\end{align}
for any triple $i,j,k\in\{0,\ldots,7\}$.
\end{prop}

\begin{proof}

By \eqref{Gamma_ratio_2}, the integrand
\begin{align}
\Delta(z,\smb{z};\vt,\smb{\vt}):=\frac{\prod_{j=0}^{7}\Gamma(\vt_j\pm z,\smb{\vt}_j\pm \tilde{\smb{z}})}{\Gamma(\pm 2z,\pm 2\tilde{\smb{z}})},
\label{integrand}
\end{align}
satisfies
\begin{align}
T_{\tau,k}\Delta(z,\smb{z};\vt,\smb{\vt})=\theta_\sigma(\vt_k\pm z,\smb{\vt}_k\pm \smb{z})\Delta(z,\smb{z};\vt,\smb{\vt}),
\end{align}
with respect to the shift operator $T_{\tau, k}$.  
Note that by the choice of $\tilde{\smb{z}}$ in \eqref{ztdef}, $\smb{z}_j\pm \tilde{\smb{z}}$, and $\pm2\tilde{\smb{z}}$, are always integers.  
Next by the three-term relation \eqref{tht13term}, the equation
\begin{align}
\begin{split}
\left(\EXP^{\frac{2\pi\ii}{r}(\vt_k-\smb{\vt}_k)}\theta_\sigma(\vt_j\pm \vt_k,\smb{\vt}_k\pm \smb{x}_k)\,T_{\tau,i}+\EXP^{\frac{2\pi\ii}{r}(\vt_i-\smb{\vt}_i)}\theta_\sigma(\vt_k\pm \vt_i,\smb{\vt}_k\pm \smb{\vt}_i)\,T_{\tau,j}+\right. \\
\left.\EXP^{\frac{2\pi\ii}{r}(\vt_j-\smb{\vt}_j)}\theta_\sigma(\vt_i\pm \vt_j,\smb{\vt}_i\pm \smb{x}_j)\,T_{\tau,k}\right)\Delta(z,\smb{z};\vt,\smb{\vt})=0,
\end{split}
\end{align}
holds for $i,j,k\in\{0,\ldots,7\}$.  By commuting the shift operator with the sum/integral, we obtain \eqref{contrel}.
\end{proof}

\subsection{\texorpdfstring{$W(E_7)$ Transformation}{W(E7) Transformation}}\label{subsec.E7}

\begin{prop}\label{prop.e7trans}
For $\vt=(\vt_{0}, \dots, \vt_{7})\in\mathbb{C}^8$, $\im(x_i)>0$, and $\smb{\vt}=(\smb{\vt}_{0}, \dots, \smb{\vt}_{7})\in \mathbb{Z}^8\cup (\mathbb{Z}+1/2)^8$ 
with the restriction 
\begin{align}
\label{balanc}
\sum_{i=0}^{7}\vt_i= 2(\sigma+\tau),\qquad\sum_{i=0}^7 \smb{\vt}_i = 2r,
\end{align}
the sum/integral \eqref{ebsi2} satisfies
\begin{align}
\label{trans3}
I(\vt,\smb{\vt})=I(\tilde{\vt},\tilde{\smb{\vt}})\prod_{0\le i<j \le 3 \textrm{ or } 4\le i<j \le 7 } \Gamma(\vt_i+\vt_j, \smb{\vt}_i+\smb{\vt}_j ),
\end{align}
where the transformed variables 
$\tilde{\vt}=(\tilde{\vt}_{0}, \dots, \tilde{\vt}_{7})\in\mathbb{C}^8$ and $\tilde{\smb{\vt}}=(\tilde{\smb{\vt}}_{0}, \dots, \tilde{\smb{\vt}}_{7})\in \mathbb{Z}^8\cup (\mathbb{Z}+\frac{1}{2})^8$
are given by
\begin{align}
\label{tilde}
\tilde{x}_i&=
\begin{cases}
x_i+\frac{\sigma+\tau}{2}-\frac{1}{2}\sum_{j=0}^3 x_j \quad & (i=0,1,2,3),\\
x_i+\frac{\sigma+\tau}{2}-\frac{1}{2}\sum_{j=4}^7 x_j \quad & (i=4,5,6,7),
\end{cases} \\
\tilde{\tilx}_i&=
\begin{cases}
\tilx_i+\frac{r}{2}-\frac{1}{2}\sum_{j=0}^3 \tilx_j \quad & (i=0,1,2,3),\\
\tilx_i+\frac{r}{2}-\frac{1}{2}\sum_{j=4}^7 \tilx_j \quad & (i=4,5,6,7).
\end{cases}
\end{align}

\end{prop}

Proposition \ref{prop.e7trans} is proven with the use of a variation of the elliptic beta sum/integral formula \cite{Kels:2015bda} in Appendix \ref{app.e7trans} (a similar proof of this identity first appeared in \cite{rarified}).

Note that the variables of the elliptic hypergeometric sum/integral \eqref{ebsi2},
essentially transform in the formula \eqref{trans3} under the action of a reflection for an element of the Weyl group $W(E_7)$.
This property is particularly important for the 
 construction of the $\uptau$-function from the sum/integral \eqref{ebsi2} (see Section \ref{sec.main}).

\section{\texorpdfstring{$\uptau$-function on the $E_8$ root lattice}{Tau-function on the E8 root lattice}}\label{sec.tau}

In this section we will consider the properties of the root lattice of $E_8$ which are used to define our $\uptau$-function.  Many of the properties and definitions are essentially based on the work of Noumi \cite{Noumi}, which in the following is related to (but not the same as) the $r=1$ case.

\subsection{\texorpdfstring{$E_8$ Root Lattice}{E8 Root Lattice}}\label{sec.E8root}

We denote the root lattice of $E_8$ by $Q(E_8)$,
and the $E_8$ Weyl group by $W(E_8)$.
The root lattice is more explicitly given as a $\mathbb{Z}$-span of the 
vectors
\begin{align}
\begin{split}
&\pm v_i \pm v_j, \quad (0\le i < j \le 7) , \\
&\frac{1}{2}\left( \pm v_0 \pm v_1 \pm \dots \pm v_7 \right), \quad (\textrm{even number of minus signs}) ,
\end{split}
\label{vectors}
\end{align}
where $\{ v_0, \dots, v_7 \}$, is the orthonormal basis with respect to the canonical
symmetric bilinear form $(-|-)$ on the root lattice $Q(E_8)$, namely
$(v_i | v_j)=\delta_{i,j}$. Note also that 
$(a|a)=2$, for $a\in \Delta(E_8)$, where $\Delta(E_8)$ is the root system for $E_8$.

The following set of vectors in $\mathbb{C}^8$ plays a central role for this paper.

\begin{definition}
A set $\{\pm a_0, \pm a_1, \dots, \pm a_{l-1} \}$ of $2l$ vectors in 
$\mathbb{C}^8$ is called a $C_l$-frame if the following two conditions are satisfied:
~\\[-0.3cm]

(1)  $(a_i|a_j)=\delta_{ij} \quad  (0\le i, j <l)$,
~\\[-0.3cm]

(2) $a_i \pm a_j \in Q(E_8) \quad (0\le i <j <l) ,  \quad 2a_i \in Q(E_8)  \quad (0\le i <l)$.

\end{definition}

Notice that this definition implies that
the set of $2l^2$ vectors
\begin{align}
\left\{\pm  a_i \pm a_j \, \big|\, (0\le i <j <l) \right\} \cup \left\{\pm 2a_i \, \big|\, (0\le i<l) \right\}
\end{align}
is contained in the root lattice $Q(E_8)$ and forms a root lattice of type $C_l$.
In the following sections we will mostly work with the $C_3$-frame for $l=3$.

\subsection{\texorpdfstring{$E_8$ $\uptau$-function }{E8 tau function}}

For a pair $Z=(z, \tilz)\in\mathbb{C}\times\mathbb{Z}$, where $z\in\mathbb{C}$, and $\tilz\in\mathbb{Z}$,
we define
\begin{align}
[Z]=[(z, \tilz)]:=\EXP^{\frac{\pi\ii}{r}(-z+\tilz)}\theta_\sigma(z ,\tilz),
\label{def_bracket}
\end{align}
where $\theta_\sigma(z ,\tilz )$ is the lens theta function defined in \eqref{lthtdef}.

The function \eqref{def_bracket} satisfies the following identities (note that these identities are simple corollaries of Proposition \ref{prop.theta}, but are written here explicitly for convenience)

\begin{prop}\label{prop.three-term}

We have the following identities for the bracket:

(1) (periodicity)
\begin{align}
[Z+(0,2r)]=[Z].
\end{align}

(2) (reflection)
\begin{align}
[-Z]= -[Z].
\end{align}

(3) (three-term identity) For $Z,Z_i,Z_j,Z_k\in \mathbb{C}\times\mathbb{Z}$, or $Z,Z_i,Z_j,Z_k\in \mathbb{C}\times(\mathbb{Z}+\frac{1}{2})$,
\begin{align}
[Z_j\pm Z_k] \,[Z_i\pm Z]+ [Z_k\pm Z_i]\,[Z_j\pm Z]+[Z_i\pm Z_j]\,[Z_k\pm Z]=0. 
\label{three-term}
\end{align}
where we used the shorthand notation $[X\pm Y]:=[X+Y][X-Y]$.
\end{prop}

Due to \eqref{thetar1}, for $r=1$ there is no dependence on the second
argument $\smb{z}$, and the bracket may simply be written as $[z]$ with $z\in \mathbb{C}$.
In that case, the three-term identity in Proposition \ref{prop.three-term},
exactly reduces to the standard three-term identity for the theta function, given in (2.1) of \cite{Noumi}. 

Consider now the space 
\begin{align}
V:=\mathbb{C}^8\times Q(E_8) .
\label{def_V}
\end{align}
The first (second) factor $\mathbb{C}^8$ ($Q(E_8)$) can be thought of as a $\mathbb{C}$-span ($\mathbb{Z}$-span)
of the root lattice generators \eqref{vectors}.  We denote an element of this space as $X=(x, \tilx)\in V$, with 
$x\in \mathbb{C}^8$, and $\tilx\in Q(E_8)$.
A natural addition on this space is defined by $(X+Y)=(x+y, \tilx+\tily)$. 

We define $\tau$ to be a non-zero complex number, and choose a region $D\subseteq V$, satisfying
\begin{align}
D= D +Q(E_8) \vT,
\label{D_condition}
\end{align}
where $T$ is the ``step size'', defined as
\begin{align}
\vT:=(\tau, 1) \label{step_T} ,
\end{align}
and we have used the notation $v \vT= (v \tau; v) \in V$, 
where $v\in Q(E_8)$.

As an example, $D$ may be chosen as the whole space $D=V$, 
as this will obviously satisfy the condition \eqref{D_condition}.
As another example, we could also minimally choose a completely discrete set for $D$, as
\begin{align}
D=C + Q(E_8) \vT ,
\end{align}
for some point $C\in V$.
Similarly to the situation in \cite{Noumi}, the construction of the hypergeometric $\uptau$-function in Section \ref{sec.E7}, will in fact
involve a combination of discrete and continuous spaces, where $D$ is chosen as an infinite family of parallel hyperplanes in $V$, that are indexed by an integer $n$.

In the following, for $a\in Q(E_8)$, and $X=(x, \tilx) \in D$, we define $(a|X)=((a|x), (a| \tilx)) \in \mathbb{C}\times \mathbb{Z}$. 

Our $\uptau$-function on $D\subset V$ is defined as follows.

\begin{definition}
A function $\uptau(X)$ defined over the region $D$ satisfying \eqref{D_condition}, is called a $\uptau$-function 
if it satisfies the non-autonomous bilinear Hirota equations 
\begin{align}
\begin{split}
&[(a_1\pm a_2 | X)] \uptau (X\pm a_0\vT) +[(a_2\pm a_0 | X)] \uptau (X\pm a_1\vT) \\
&\qquad+[(a_0\pm a_1 | X)] \uptau (X\pm a_2\vT)=0 ,
\end{split}
\label{Hirota}
\end{align}
for any $C_3$-frame $(a_0, a_1, a_2)$, and $X\in D$.
\end{definition}

For a general choice of $D\subseteq V$ satisfying \eqref{D_condition}, even for $r=1$ the Hirota equations \eqref{Hirota}  will have a non-trivial dependence on the discrete variables coming from $Q(E_8)$.  Indeed the hypergeometric solution of \eqref{Hirota} obtained in Section \ref{sec.main}, will have such a dependence for all $r=1,2,\ldots $. In this respect, the situation considered here is a different situation than was considered in \cite{Noumi}, where $\uptau$-functions in the latter were defined on subsets of $V=\mathbb{C}^8$, and have no dependence on any discrete variables.

As an example, note that the Hirota equations \eqref{Hirota} admit the following constant solution:

\begin{prop}
For $X=(x, \tilx)\in V$, and a constant $C=(c, \tilc)\in V$, the function
\begin{align}
\uptau(X)=\left[\left( \frac{1}{2 \tau} (x|x) +c,  \frac{1}{2} (\tilx | \tilx)  + \tilc  \right) \right] \quad (X\in V),
\end{align}
is an example of a $\uptau$-function associated with the region $D=V$.

\end{prop}

\begin{proof}
We have
\begin{align}
\begin{split}
\uptau(X\pm a_i \vT)&=\left[ \left( \frac{1}{2 \tau} (x|x) +\tau+ c \pm (a_i|x) ,  \frac{1}{2} (\tilx | \tilx)  + 1 + \tilc \pm (a_i|\tilx)\right) \right]\\
&=\left[ Z \pm Z_i \right],
\end{split}
\end{align}
where 
\begin{align}
\begin{split}
Z&:=\left(\frac{1}{2 \tau} (x|x) +\tau+c, \frac{1}{2} (\tilx | \tilx)  +1+ \tilc \right) , \\
Z_i&:=(a_i|X)=\left( (a_i|x), (a_i|\tilx)\right)\quad (i=0,1,2) .
\end{split}
\end{align}
The Hirota equations \eqref{Hirota} then follow from the three-term identity \eqref{three-term}.
\end{proof}

For a given $\uptau$-function, one can also construct a new $\uptau$-function
by an element of the Weyl group $W(E_8)$. In this sense the Hirota equations are 
``covariant'' with respect to the action of $W(E_8)$:

\begin{prop} For a $\uptau$-function $\uptau$ on a domain $D$, 
and an element $w\in W(E_8)$, the function $w \cdot \tau$ defined by
\begin{align}
(w\cdot \uptau)(X): =\uptau(w^{-1} \cdot X)  \quad (X\in w \cdot D),
\end{align}
is also a $\uptau$-function on the domain $w \cdot D$.  

\end{prop}

Note that $w\cdot \tau$ is in general different from $\tau$, particularly they will respectively be defined on different domains.

\section{\texorpdfstring{Decomposition into $E_7$-Orbits}{Decomposition into E7-Orbits}}\label{sec.E7}

\subsection{\texorpdfstring{Decomposition of $D$}{Decomposition of D}}

In the previous section we have considered the domain $D$ of the $\uptau$-function, as a general
subset of $V$, satisfying the condition \eqref{D_condition}.  To start to consider hypergeometric solutions, we proceed with a special choice of $D$, given by
\begin{align}
\begin{split}
&D=\bigsqcup_{n\in \mathbb{Z}} D_n  ,\\
&D_n=H_{n \tau+\sigma} \times  \left( H_{n+r-1} \cap   Q(E_8) \right) ,
\end{split}
\label{D_H}
\end{align}
where the hyperplane $H_{\kappa}$, is defined by 
\begin{align}
H_{\kappa}= \left\{ x \in \mathbb{C}^8 \big|\, (x|\phi)= \kappa
\right\},
\end{align}
and $\phi$ is the highest root, which in the basis $v_0, \dots, v_7$ of \eqref{vectors}, is
given by 
\begin{align}
\phi=\frac{v_0+ \dots +v_7}{2}.
\end{align}
Thus the coordinates $X=(x,\smb{x})\in D_n$ satisfy
\begin{align}
\sum_{i=0}^7 x_i=2 (n\tau+\sigma), \quad
\sum_{i=0}^7 \tilx_i= 2(n+r-1) .
\end{align}

The choice of the highest root $\phi$  breaks the manifest covariance
under the $W(E_8)$ symmetry down to the stabilizer of $\phi$,
which is the Weyl group $W(E_7)$. 
Indeed, the $E_7$ root lattice in the basis of \eqref{vectors},
is spanned by 
\begin{align}
\begin{split}
&\pm (v_i - v_j),  \quad (0\le i < j \le 7)  ,\\
&\frac{1}{2}\left( \pm v_0 \pm v_1 \pm \dots \pm v_7 \right), \quad (\textrm{total of four minus signs}),
\end{split}
\label{vectors_2}
\end{align}
and these vectors together with the highest root $\phi$ of $E_8$, generate the 
whole $E_8$ root lattice.

\begin{rem}
Recall that there were 2 types of lens theta functions $\theta_\sigma$, $\theta_\tau$, defined in \eqref{lthtdef}, while the bracket function \eqref{def_bracket} is defined in terms of $\theta_\sigma$ only.  However, in the definition of the bracket function \eqref{def_bracket}, we could also replace $\theta_\sigma$ with $\theta_\tau$, as:
\begin{align}
\overline{[Z]}=\overline{[(z, \tilz)]}:=\EXP^{\frac{\pi\ii}{r}(-z_j+\tilz_j)}\theta_\tau(z_j ,\tilz_j ) .
\label{def_bracket_other}
\end{align}
This bracket will still satisfy Proposition \ref{prop.three-term}, from which we can build the  
lens-elliptic $\uptau$-function.  Then in this case, instead of \eqref{D_H}, the definition of a suitable region $D$
 would be 
\begin{align}
\begin{split}
&D=\bigsqcup_{n\in \mathbb{Z}} D_n ,  \\
&D_n=\left( H_{n \sigma +\tau} \oplus  \left( H_{-n+r+1} \cap   Q(E_8) \right) \right).
\end{split}
\label{D_H_other}
\end{align}
\end{rem}

\subsection{\texorpdfstring{Decomposition of $\uptau$-function}{Decomposition of tau-function}}

Let us now analyse the $\uptau$-function on the domain $D$ given in  \eqref{D_H}. 
Since \eqref{D_H} is a disjoint union, the $\uptau$-function on the domain $D$, can be thought of as 
an infinite sequence of functions $\uptau^{(n)}$ on $D_n$, which are indexed by the integer $n$:
\begin{align}
\uptau^{(n)}:=\uptau|_{D_{n}} .
\end{align}
We wish to write the Hirota equations \eqref{Hirota} as a set of conditions for the $\uptau^{(n)}$ defined on $D_n$.

In the Hirota equation \eqref{Hirota}, the argument $X$ is shifted by vectors $\pm a_0$, $\pm a_1$, $\pm a_2$
which come from the particular choice of $C_3$-frame. This means that the corresponding Hirota equations on $D$, will provide relations between $\uptau^{(n)}$-functions on up to three different hyperplanes, depending on the values of $(\phi|a_{i})_{i=0,1,2}$, for the particular $C_3$-frame.  In terms of the inner product $(\phi|a_{i})_{i=0,1,2}$, the $C_3$-frames may be classified as one of the following four types:

\begin{prop}[Proposition 3.2 in \cite{Noumi}] \label{prop.C3_frame}
The set of all $C_3$-frames may be decomposed into four $W(E_7)$-orbits. 
For $\{ \pm a_0, \pm a_1, \pm a_2\}$, the orbit is classified as one of the four types $(I)$, $(II_0)$, $(II_1)$, $(II_2)$, according to the pairings with the highest root $\phi$:
\begin{align}
\begin{split}
&(\I):\quad (\phi|a_0)=(\phi|a_1)=(\phi|a_2)=\frac{1}{2} , \\
&(\II_0):\quad (\phi|a_0)=(\phi|a_1)=(\phi|a_2)=0 ,\\
&(\II_1):\quad (\phi|a_0)=1, \quad (\phi|a_1)=(\phi|a_2)=0, \\
&(\II_2):\quad (\phi|a_0)=(\phi|a_1)=1, \quad (\phi|a_2)=0.
\end{split}
\label{phi_a}
\end{align}
\end{prop}

\begin{rem}\label{rem.enlarge_C3}
The notation $(\I), (\II_0), (\II_1), (\II_2)$ is motivated by the facts that (see \cite{Noumi}, Propositions 1.4 and 3.1)

(1) any $C_3$-frame is contained in a unique $C_8$-frame.

(2) The set of $C_8$ frames may be decomposed into two $W(E_7)$-orbits,
which are characterized by
\begin{align}
&(\I):\quad (\phi|a_i)=\frac{1}{2}, \quad (i=0,\dots, 7),  \\
&(\II):\quad (\phi|a_0)=(\phi|a_1)=1 ,  \quad (\phi|a_i)=0 , \quad (i=2,\dots, 7).
\end{align}
Moreover, in the case of $(\II)$, we can show that $a_0+a_1=\phi$.

\end{rem}

This remark implies that a given $C_3$-frame can be enlarged nicely into a $C_8$-frame.
As an example, suppose that we have a $C_3$-frame $\{\pm a_0, \pm a_1, \pm a_2\}$
of type  $(\II_2)$. We can then choose a $C_8$-frame 
$\{\pm a_0, \pm a_1, \dots, \pm a_7\}$ of type $(\II)$ containing the $C_3$-frame $\{\pm a_0, \pm a_1, \pm a_2\}$ that we started with.
This $C_8$-frame also contains many other $C_3$-frames---for example $\{\pm a_1, \pm a_2, \pm a_3\}$, as a $C_3$-frame of type $(\II_1)$. 
This type of manipulation will be useful for some of the proofs below.

\medskip

Thanks to Proposition \ref{prop.C3_frame}, we find that 
there are four different types of Hirota identities depending on the different types of $C_3$-frames.
These are, for $X\in D_{n+\frac{1}{2}}$,
\begin{align}
\begin{split}
& [(a_1\pm a_2 | X)] \uptau^{(n)} (X- a_0\vT) \uptau^{(n+1)} (X+ a_0\vT)  \\
({\rm I})_{n+\frac{1}{2}}: \quad& \qquad +[(a_2\pm a_0 | X)]  \uptau^{(n)} (X- a_1\vT)\uptau^{(n+1)} (X+ a_1\vT) \\
& \qquad+[(a_0\pm a_1 | X)] \uptau^{(n)} (X- a_2\vT) \uptau^{(n+1)} (X+ a_2\vT)=0,
\end{split}
\end{align}
and for $X\in D_{n}$, 
\begin{align}
\label{II_0_n}
\begin{split}
 &[(a_1\pm a_2 | X)] \uptau^{(n)} (X\pm a_0\vT)  \\
({\rm II}_0)_{n}: \quad & \qquad +[(a_2\pm a_0 | X)]  \uptau^{(n)} (X \pm  a_1\vT) \\
&\qquad +[(a_0\pm a_1 | X)] \uptau^{(n)} (X\pm a_2\vT) =0 ,
\end{split}
\\
\label{II_1}
\begin{split}
& [(a_1\pm a_2 | X)] \uptau^{(n-1)} (X- a_0\vT) \uptau^{(n+1)} (X+ a_0\vT) \\
({\rm II}_1)_{n}: \quad& \qquad+[(a_2\pm a_0 | X)]  \uptau^{(n)} (X\pm a_1\vT)  \\
& \qquad+[(a_0\pm a_1 | X)] \uptau^{(n)} (X\pm  a_2\vT) =0 ,
\end{split}
\\
\label{II_2}
\begin{split}
& [(a_1\pm a_2 | X)] \uptau^{(n-1)} (X- a_0\vT) \uptau^{(n+1)} (X+ a_0\vT) \\
({\rm II}_2)_{n}: \quad & \qquad+[(a_2\pm a_0 | X)]  \uptau^{(n-1)} (X- a_1\vT)\uptau^{(n+1)} (X+ a_1\vT) \\
&\qquad +[(a_0\pm a_1 | X)]  \uptau^{(n)} (X\pm a_2\vT)=0 .
\end{split}
\end{align}

Thus we have decomposed the Hirota equations \eqref{Hirota} for the $\uptau$-function of type $E_8$ on $D$,
into a set of equations for an infinite sequence of $\uptau^{(n)}$-functions satisfying Hirota equations $(II_0)_n$ of type $E_7$ on the $D_n$.

Furthermore, the four identities above are not independent, and 
in fact we can focus on $(\II_1)_n$ only, from which all others can be derived:

\begin{prop}\label{prop.follows}
For $n\in \mathbb{Z}$
we can derive $(\I_1)_{n+\frac{1}{2}},  (\II_0)_{n+1}$ and $(\II_2)_{n}$
from $(\II_1)_n$:
\begin{align}
&(\II_1)_n \Longrightarrow (\I)_{n+\frac{1}{2}} \\
&(\II_1)_n \Longrightarrow (\II_0)_{n+1} \\
&(\II_1)_n \Longrightarrow (\II_2)_{n} 
\end{align}
\end{prop}

\begin{proof} 
Let us here prove only the first statement ($(\II_1)_n \Longrightarrow (\II_2)_{n})$, since the argument is similar for other cases
(see also \cite[Appendix A]{Noumi} and \cite[section 3]{Masuda}).

We wish to show $(\II_2)_n$ for a $C_3$-frame of type $(\II_1)$,
namely the set $\{ \pm a_0, \pm a_1, \pm a_2 \}$ satisfying
$(\II_2):\, (\phi|a_0)=(\phi|a_1)=1, \quad (\phi|a_2)=0$.
We can choose one more element $a_3$ from the root lattice,
such that $(\phi|a_3)=0$ (see Remark \ref{rem.enlarge_C3}). Then the sets
$\{ \pm a_0, \pm a_2, \pm a_3 \}$ and $\{ \pm a_1, \pm a_2, \pm a_3 \}$
are $C_3$-frames of type $(\II_1)$ respectively.
From the assumption of $(\II_1)_n$ 
we have
\begin{align}
& [(a_2\pm a_3 | X)] \uptau^{(n-1)} (X- a_i\vT) \uptau^{(n+1)} (X+ a_i\vT) \nonumber \\
& \qquad+[(a_3\pm a_i | X)]  \uptau^{(n)} (X\pm a_2\vT) +[(a_i\pm a_2 | X)] \uptau^{(n)} (X\pm  a_3\vT) =0 . \nonumber \\
\end{align}
for $i=0,1$. We can use these equations to compute
\begin{align}
\begin{split}
& [(a_1\pm a_2 | X)] \uptau^{(n-1)} (X- a_0\vT) \uptau^{(n+1)} (X+ a_0\vT) \\
&\qquad +[(a_2\pm a_0 | X)]  \uptau^{(n-1)} (X- a_1\vT)\uptau^{(n+1)} (X+ a_1\vT) \\
& \qquad+[(a_0\pm a_1 | X)]  \uptau^{(n)} (X\pm a_2\vT) \\
=
&- \frac{[(a_1\pm a_2 | X)]}{[(a_2\pm a_3 | X)]}  \left( [(a_3\pm a_0 | X)]  \uptau^{(n)} (X\pm a_2\vT)+[(a_0\pm a_2 | X)] \uptau^{(n)} (X\pm  a_3\vT) \right)  \\
&- \frac{[(a_2\pm a_0 | X)]}{[(a_2\pm a_3 | X)]}  \left( [(a_3\pm a_1 | X)]  \uptau^{(n)} (X\pm a_2\vT)+[(a_1\pm a_2 | X)] \uptau^{(n)} (X\pm  a_3\vT) \right)  \\
& +[(a_0\pm a_1 | X)]  \uptau^{(n)} (X\pm a_2\vT)\\
=
& -\frac{\uptau^{(n)} (X\pm a_2\vT)}{[(a_2\pm a_3 | X)]}
\left(
[(a_1\pm a_2 | X)] [(a_3\pm a_0 | X)]  +
[(a_2\pm a_0 | X)]   [(a_3\pm a_1 | X)] \right. \\
& \qquad \qquad\qquad\qquad\qquad\left.
-[(a_0\pm a_1 | X)][(a_2\pm a_3 | X)]
\right) .
\end{split}
\end{align}
This vanishes thanks to the three-term identity \eqref{three-term}. We have therefore proven $(\II_2)_{n}$.  The cases of $(\I)_{n+\frac{1}{2}}$, and  $(\II_0)_{n+1}$ are similar.
\end{proof}

\section{\texorpdfstring{Hypergeometric $\uptau$-Function}{Hypergeometric tau-Function}}\label{sec.main}

\subsection{Main Theorem}

In this section we give explicit lens elliptic hypergeometric solutions for the  $E_8$ $\uptau$-function on $D$,
as an infinite sequence of $E_7$ $\uptau^{(n)}$-functions on $D_n$.


\begin{definition}
A $\uptau$-function on $D$, with $\uptau^{(n)}=0$ for $n<0$, is called hypergeometric.
\end{definition}

The hypergeometric solution may be expressed in either a determinant form, and a multi-dimensional sum/integral form, and the latter two forms are equivalent to each other.

To state our main theorem it will be convenient to first define some additional functions.  
First we define the function $\psi_{ij}^{(n)}(X)$ $(X\in D_n)$, as 
\begin{align}
&\psi^{(n)}_{ij}(X):=\psi\left(X+v_{ij}^{(n)}\vT \right), \label{psi_n}\\
&v_{ij}^{(n)}:=(1-n)a_0+(n+1-i-j)a_1+(j-i)a_2, \label{v_n}
\end{align}
where $\{ \pm a_0, \pm a_1, \pm a_2 \}$ is a $C_3$-frame (which we fix for the moment),
and $\psi$ is given in terms of the elliptic hypergeometric sum/integral \eqref{ebsi2}, as 
\begin{align}
\psi(X)=I(\tilde{x},\tilde{\tilx}; \sigma,\tau),
\end{align}
where the transformed variables $\tilde{x}$ and $\tilde{\tilx}$ are defined in \eqref{tilde}.

Next, we define a function $\mathscr{G}^{(n)}$ in terms of the lens triple gamma functions \eqref{Gam3def}, by
\begin{align}
\begin{split}
\mathscr{G}^{(n)}(X)& := \hspace{-0.2cm}
\prod_{\substack{{0\le i \le 3 } \\ { 4\le j \le 7}}} \left\{\Gt_{\sigma}\left((1-n) \tau+ x_i+x_j, 1-n +\tilx_i+\tilx_j ; \sigma, \tau, \tau\right)\right. \\[-0.4cm]
&\qquad\phantom{\prod_{\genfrac{}{}{0pt}{1}{0\le i<j \le 3 }{\textrm{or } 4\le i<j \le 7}}}\times\left.\left(\gt_\tau(x_i+x_j;,\tilx_i+\tilx_j;\sigma,\tau,\mu)\right)^n\right\} \\[-0.3cm]
&\times \prod_{\substack{{0\le i<j \le 3 } \\ { \textrm{or } 4\le i<j \le 7}}} \left\{\Gt_\sigma(\tau+ x_i+x_j,1+ \tilx_i+\tilx_j ; \sigma, \tau, \tau)\right. \\[-0.4cm]
&\qquad\phantom{\prod_{\genfrac{}{}{0pt}{1}{0\le i<j \le 3 }{\textrm{or } 4\le i<j \le 7}}}\times\left.\left(\gt_\tau(x_i+x_j+\mu;,\tilx_i+\tilx_j;\sigma,\tau,\mu)\right)^n\right\} .
\end{split}
\label{G_def}
\end{align}

We define a function $d^{(n)}$ by 
\begin{align}
\begin{split}
d^{(n)}(X)&:=\EXP^{\frac{4\pi\ii}{r}(\tau-1)\binom{n}{3}}\EXP^{\frac{2\pi\ii}{r}(\sigma+\tau-x_0-x_1+\tilx_0+\tilx_1)\binom{n}{2}} \\
&\times\prod_{k=1}^n  \Big[ \theta_\sigma(x_0+x_3+(1-n)\tau,\tilx_0+\tilx_3+(1-n))_{k-1} \\[-0.2cm]
&\qquad\quad\theta_\sigma(x_1+x_2+(1-n)\tau,\tilx_1+\tilx_2+(1-n))_{k-1} \\
&\qquad\quad\theta_\sigma(x_0-x_3-(k-1)\tau,\tilx_0-\tilx_3-(k-1))_{n-k} \\
&\qquad\quad\theta_\sigma(x_1-x_2-(k-1)\tau,\tilx_1-\tilx_2-(k-1))_{n-k} \Big],
\end{split}
\label{d_def}
\end{align}
where 
\begin{align}\label{theta_k}
\theta_\sigma(x,\tilx)_k:=
\prod_{j=0}^{k-1} \theta_\sigma(x+j\tau,\tilx+j ), \quad (k=0,1,\dots),
\end{align}
and $\theta_\sigma$ is a lens theta function defined in \eqref{lthtdef}.

We also define $e^{(n)}$ by
\begin{align}
e^{(n)}(X)& :=\EXP^{\frac{2\pi\ii}{r}(\sigma+1) \binom{n}{2}} e^{-n Q(X)} , \label{e_def}
\end{align}
where
\begin{align}
Q(X)&:=\frac{2\pi \ii}{r}\left( \frac{1}{2 \tau}(x|x)- \frac{1}{2}(\tilx|\tilx)  \right). \label{Q_def}
\end{align}

Finally, we define a function $g^{(n)}(X)$, as the following combination of the above three functions
\begin{align}
g^{(n)}(X)& :=\frac{e^{(n)}(X)}{d^{(n)}(X)}
\mathscr{G}^{(n)}(X) .
\label{g_def}
\end{align}

We now come to the main theorem of this paper:
\begin{thm} \label{theorem.tau_n}

For a
$C_3$-frame $\{ \pm a_0, \pm a_1, \pm a_2 \}$ of type $\II_1$, the function $\uptau(X)$ on $D=\sqcup_{n\in \mathbb{Z}} D_n$ \eqref{D_H}, defined on each $D_n$ by
\begin{empheq}[box={\mybluebox[7pt]}]{align}
\uptau^{(n)}(X) =  \uptau(X) \big|_{D_n}:=
\begin{cases} g^{(n)}(X) \det \left(\psi_{i,j}^{(n)}(X)\right)_{i,j=1}^n & (n\ge 0), \\
0 &  (n<0) ,
\end{cases}
\label{def_tau_n}
\end{empheq}
is a hypergeometric $\uptau$-function which satisfies the Hirota equations \eqref{Hirota}. 
Moreover each $\uptau^{(n)}$ is invariant under the action of $W(E_7)$.

\end{thm}

\begin{rem}
The definition of the $\uptau$-function in Theorem \ref{theorem.tau_n},
a priori depends on the choice of the $C_3$-frame  of type $\II_1$. However, the $W(E_7)$-invariance,
together with the fact that any two $C_3$-frames  of type $\II_1$ are related by an element of $W(E_7)$ 
(Proposition \ref{prop.C3_frame}), means that the $\uptau$-function is actually independent of such a choice.
\end{rem}

The $n=0$, and $n=1$ cases of Theorem \ref{theorem.tau_n} are explicitly given  by:
\begin{align}
\uptau^{(0)}(X)&=g^{(0)}(X)=\prod_{0\le i<j \le 7} \Gt_\sigma(\tau+ x_i+x_j, 1+ \tilx_i+\tilx_j ; \sigma, \tau, \tau)  , \label{tau_0} 
\end{align}
and
\begin{align}
\begin{split}
&\uptau^{(1)}(X)=g^{(1)}(X) \,\psi(X)\\
&=e^{-Q(X)}I(\tilde{x},\tilde{\tilx},\sigma,\tau) \prod_{\substack{{0\le i \le 3 } \\ { 4\le j \le 7}}} \Gt_{\sigma\tau}(x_i+x_j, \tilx_i+\tilx_j ; \sigma, \tau, \tau,\mu) \\
& \times \hspace{-0.1cm} \prod_{\substack{0\le i<j \le 3 \\ \textrm{ or } 4\le i<j \le 7}} \hspace{-0.1cm} \Gt_\sigma( \tau+x_i+x_j,1+ \tilx_i+\tilx_j ; \sigma, \tau, \tau)\, \gt_\tau(x_i+x_j+\mu;,\tilx_i+\tilx_j;\sigma,\tau,\mu) 
  \\
&=e^{-Q(X)}I(x, \tilx; \sigma, \tau) \prod_{0\le i<j\le 7} \Gt_{\sigma\tau}( x_i+x_j, \tilx_i+\tilx_j ; \sigma, \tau, \tau,\mu) ,
\end{split}
\label{tau_1}
\end{align}
where in the last line we have used the transformation \eqref{trans3}.  This last equality of \eqref{tau_1}  gives a manifestly $\mathfrak{S}_8$-symmetric expression for $\tau^{(1)}$.

\begin{prop} \label{rem.recursion}
The $\uptau$-function given by Theorem \ref{theorem.tau_n} is non-zero for all $n\geq 0$.
It is also the unique hypergeometric $\uptau$-function for the initial conditions \eqref{tau_0}, \eqref{tau_1}.
\end{prop}

\begin{proof}

Suppose that $\uptau^{(n+1)}\equiv 0$ for some $n>0$.
Then from the recurrence relation \eqref{II_2}
we find $\uptau^{(n)} (X\pm a_2\vT)=0$. By repeating this recurrence we arrive at
$\uptau^{(1)}=0$, which contradicts the expression for $\uptau^{(1)}$ given in \eqref{tau_1}. 
Therefore $\uptau^{(n)}\ne 0$ for all $n\geq 0$.

Next, the recurrence relation \eqref{II_1} means that we can recursively determine $\uptau^{(n+1)}$ $(n>1)$  from
$\uptau^{(n)}$ and $\uptau^{(n-1)}$.
This in turn means that the $\uptau$-function in Theorem \ref{theorem.tau_n} is 
a unique hypergeometric function with given 
$\uptau^{(0)}$ and $\uptau^{(1)}$ (\it{cf.}\ Theorem 4.2 of \cite{Noumi}). 
\end{proof}

\subsection{Multi-Dimensional Sum/Integral Expression}
In the definition \eqref{def_tau_n},
the $\uptau$-function $\uptau^{(n)}$
was expressed in terms of an $n\times n$ determinant. The $\uptau$-function also has
the following equivalent expression given in terms of an $n$-dimensional 
elliptic hypergeometric sum/integral:

\begin{thm}\label{theorem.tauint}
The $\uptau$-function
$\uptau^{(n)}$ on  $D_n$, can be written in terms of a multi-dimensional elliptic hypergeometric sum/integral, as
\begin{align}
\label{tau_n_product}
\tau^{(n)}(X)=g^{(n)}(X) \, d^{(n)}(X) \, I_n(X)= e^{(n)}(X) \, \mathscr{G}^{(n)}(X) \, I_n(X),
\end{align}
where
\begin{align}
\label{tau_n_integral}
I_n(X)=\frac{\lambda^n}{n!}\sum_{\smb{z}_1,\ldots,\smb{z}_n=0}^{r-1}\int_{[0,1]^n}\prod_{1\leq i<j\leq n}\theta_\sigma(\pm z_i\pm z_j,\pm \hat{\smb{z}}_i\pm \hat{\smb{z}}_j)\prod_{k=1}^n\Delta(z_k, \hat{\smb{z}}_k; \tilde{x}, \tilde{\tilx})dz_k,
\end{align}
$\Delta(z,\smb{z}; \tilde{x}, \tilde{\tilx})$ is the integrand defined in \eqref{integrand}
with transformed variables $(\tilde{x}, \tilde{\tilx})$ \eqref{tilde}, $\lambda$ is given by \eqref{lambda_def}, and 
\begin{align}
\hat{\smb{z}_i}: =\smb{z}_i+(r+(r+1\bmod2))(\tilde{\smb{x}}_1\bmod 1),\quad i=1,\ldots,n.
\label{hat_def}
\end{align}

\end{thm}

Note that \eqref{hat_def} takes the following values,
\begin{align}
\hat{\smb{z} }_i=
\begin{cases}
\smb{z}_i&  (\tilde{\smb{x}} \in \mathbb{Z}^8),\\
\smb{z}_i+\frac{r}{2}&  (\tilde{\smb{x}} \in \left(\mathbb{Z}+\frac{1}{2}\right)^8, \, r \textrm{ odd}), \\
\smb{z}_i+\frac{r+1}{2} & (\tilde{\smb{x}} \in \left(\mathbb{Z}+\frac{1}{2}\right)^8, \, r \textrm{ even}).
\end{cases}
\end{align}

To prove Theorem \ref{theorem.tauint}, we will use the following $r\geq1$ analogue of Warnaar's elliptic Krattenthaler determinant formula \cite{MR2093076}.

\begin{lem}\label{lem.det}
For complex $z_i$, $i=1,\ldots,n$, integer $\smb{z}_i$, $i=1,\ldots,n$, complex parameters $x_1,x_2$, and integer parameters $\smb{x}_1,\smb{x}_2$,
\begin{align}
\begin{split}
&\det\left(\theta_\sigma(x_1\pm z_i,\smb{x}_1\pm \smb{z}_i)_{j-1}\theta_\sigma(x_2\pm z_i,\smb{x}_2\pm \smb{z}_i)_{n-j}\right)^n_{i,j=1} \\
& \qquad =\EXP^{\frac{2\pi\ii}{r}(\tau-1) \binom{n}{3}}\EXP^{\frac{2\pi\ii}{r}(x_1-\smb{x}_1)\binom{n}{2}}\\
&\qquad \times\prod_{k=1}^n\theta_\sigma(x_2\pm(x_1+(k-1)\tau),\smb{x}_2\pm(\smb{x}_1+(k-1)))_{n-k} \\
&\qquad\times\prod_{1\leq i<j\leq n}\EXP^{-\frac{2\pi\ii}{r}(z_i-\smb{z}_i)}\theta_\sigma(z_i\pm z_j,\smb{z}_i\pm \smb{z}_j).
\end{split}
\end{align}
\end{lem}

\begin{proof}
This follows from the analogous identity with the regular theta functions for $r=1$ given in \cite{MR2093076}.
\end{proof}

\begin{proof}[Proof of Theorem \ref{theorem.tauint}] 

Let us compute the determinant of $\psi_{ij}^{(n)}(X)=\psi(X+v_{ij}^{(n)}\vT)$.
We choose a $C_3$-frame of type $\II_1$ as
\begin{align}
\begin{split}
a_0 &=\frac{1}{2}(v_0+v_1+v_2+v_3) ,\\
a_1 &=\frac{1}{2}(v_0+v_1-v_2-v_3) ,\\
a_2 &=\frac{1}{2}(v_0-v_1+v_2-v_3) ,
\end{split}
\label{II1_frame}
\end{align}
so that 
\begin{align}
v_{ij}^{(n)}=(1-i)v_0+(1-j)v_1+(j-n)v_2+(i-n)v_3 .
\end{align}
The dependence on $i,j$, shifts the variable $X=(x,\smb{x})$,
and when converted into the transformed variables
$\tilde{X}=(\tilde{x}, \tilde{\tilx})$ \eqref{tilde}, 
this amounts to the shift 
$\tilde{X} \to \tilde{X} + \tilde{v}_{ij}^{(n)} T$
with
\begin{align}
\tilde{v}_{ij}^{(n)}=v_{ij}^{(n)} + (n-1) \sum_{i=0}^3 v_i =(n-i)v_0+(n-j)v_1+(j-1)v_2+(i-1)v_3 .
\end{align}

From the definition \eqref{psi_n}
it follows that 
\begin{align}
\label{psicalc}
\psi_{ij}^{(n)}(X)=\psi(X+v_{ij}^{(n)}\vT)=\lambda\sum_{\smb{z}=0}^{r-1}\int_{[0,1]}dz\, \Delta(z, \hat{\smb{z}} ; \tilde{x}, \tilde{\tilx})f_i(z, \hat{\smb{z}})g_j(z, \hat{\smb{z}}),
\end{align}
where
\begin{align}
f_i(z,\smb{z})&:=\theta_\sigma(\tilde{x}_0\pm z,\tilx_0\pm \smb{z})_{n-i}\theta_\sigma(\tilde{x}_3\pm z,\tilx_3\pm \smb{z})_{i-1}, \\
g_j(z,\smb{z})&:=\theta_\sigma(\tilde{x}_1\pm z, \tilx_1\pm \smb{z})_{n-j}\theta_\sigma(\tilde{x}_2\pm z, \tilx_2\pm \smb{z})_{j-1},
\end{align}
for $i,j=1,2,\ldots, n$, where $\theta_\sigma(z,\smb{z})_k$ and $\lambda$, are defined in \eqref{theta_k}, and \eqref{lambda_def},
respectively.

The determinant of $\psi^{(n)}_{ij}(x)$ \eqref{psicalc}, may be written as
\begin{align}
\begin{split}
&\det(\psi^{(n)}_{ij}(x))^n_{i,j=1} \\
&=\frac{\lambda^n}{n!}\sum_{\smb{z}_1,\ldots,\smb{z}_n=0}^{r-1}\int_{[0,1]^n}\hspace{-0.3cm}\det(f_i(z_j, \hat{\smb{z}}_j))_{i,j=1}^n\det(g_i(z_j, \hat{\smb{z}}_j))_{i,j=1}^n
\prod_{k=1}^n\Delta(z_k,\hat{\smb{z}}_k; \tilde{x}, \hat{\tilde{\tilx}})dz_k.
\end{split}
\end{align}
This determinant may be evaluated using Lemma \eqref{lem.det}, which results in the expression \eqref{tau_n_product}, where $d^{(n)}(X)$ defined in terms of transformed variables $(\tilde{x}, \tilde{\tilx})$, is given by (equivalent to previous definition given in \eqref{d_def}), 
\begin{align}\label{d_tilde}
\begin{split}
d^{(n)}(X)&=\EXP^{\frac{4\pi\ii}{r} (\tau-1) \binom{n}{3}}\EXP^{\frac{2\pi\ii}{r}(\tilde{x}_0+\tilde{x}_1-\tilde{\tilx}_0-\tilde{\tilx}_1)\binom{n}{2}} \\
&\qquad\times\prod_{k=1}^n\theta_\sigma(\tilde{x}_0\pm(\tilde{x}_3+(k-1)\tau),\tilde{\tilx}_0\pm(\tilde{\tilx}_3+(k-1)))_{n-k} \\
&\qquad \qquad \quad  \theta_\sigma(\tilde{x}_1\pm(\tilde{x}_2+(k-1)\tau),\tilde{\tilx}_1\pm(\tilde{\tilx}_2+(k-1)))_{n-k}.
\end{split}
\end{align}
Here we have used the properties of the lens theta function $\theta_\sigma$ which appear in \eqref{thtinv} and \eqref{thtshft}.
\end{proof}

\section{Proof of Theorem \ref{theorem.tau_n}}\label{sec.proof}

In this final section we will prove Theorem \ref{theorem.tau_n}.  We begin by proving the $W(E_7)$-invariance of the $\uptau$-function of Theorem \ref{theorem.tau_n}.

\subsection{\texorpdfstring{$W(E_7)$ Invariance}{W(E7) Invariance}}

\begin{prop}\label{prop.tau0_E7}
$\uptau^{(0)}$ is $W(E_7)$-invariant.
\end{prop}

\begin{proof}

The Weyl group $W(E_7)$ is generated by Weyl reflection with respect to the roots listed in 
\eqref{vectors_2}. The Weyl reflections with respect to the roots
$\pm (v_i - v_j)$, generate the symmetric group $\mathfrak{S}_8$,
under which the $\uptau^{(0)}$ given in \eqref{tau_0} is manifestly invariant.
The remaining roots $\left( \pm v_0 \pm v_1 \pm \dots \pm v_7 \right)/2$
(with four minus signs), are mapped to each other under the symmetric group $\mathfrak{S}_8$,
and we conclude the $W(E_7)$ is generated by $\mathfrak{S}_8$ together with
an extra element $w_0 \in W(E_7)$, representing the Weyl reflection with respect to 
$\left( -v_0 -v_1-v_2-v_3+v_4+v_5+v_6+v_7\right)/2$.

Let us consider $\uptau^{(0)}$.
The Weyl reflection  $w_0 \in W(E_7)$ acts on the coordinate $X=(x, \tilx) \in D_0$  as
\begin{align}
w_0(x_i+x_j)=
\begin{cases} 
\sigma-x_i-x_j & (\{ i,j,k,l \}=\{0,1,2,3\} \textrm{ or } \{4,5,6,7\} ) ,\\
x_i+x_j & (i \in \{0,1,2,3\}, j\in \{4,5,6,7\}) , \\
\end{cases}
\end{align}
for continuous variables
and 
\begin{align}
w_0(\tilx_i+ \tilx_j)=
\begin{cases} 
r-1-\tilx_i-\tilx_j & (\{ i,j,k,l \}=\{0,1,2,3\} \textrm{ or } \{4,5,6,7\} ), \\
\tilx_i+\tilx_j & (i \in \{0,1,2,3\}, j\in \{4,5,6,7\}) , \\
\end{cases}
\end{align}
for discrete variables. Note that the constants $\sigma$, and $r-1$, in the two equations, come from the fact that we are considering the specific hyperplane $D_0\subset H_{\sigma}\times H_{r-1}$.

We  find for $\{ i,j,k,l \}=\{0,1,2,3\} \textrm{ or } \{4,5,6,7\}$, that
\begin{align}
\begin{split}
&w_0\left(\Gt_\sigma(\tau+x_i+x_j, 1+ \tilx_i+\tilx_j ; \sigma, \tau, \tau) \right)\\
&\quad =\Gt_\sigma(\tau+\sigma-x_i-x_j, r- \tilx_i- \tilx_j; \sigma, \tau, \tau) \\
&\quad =\Gt_\sigma(\tau+x_i+x_j , 1+\tilx_i+\tilx_j; \sigma, \tau, \tau),
\end{split}
\end{align}
where in the last line we used \eqref{invert_Gamma3}.
This shows the invariance of $\uptau^{(0)}$ in \eqref{tau_0} under the Weyl reflection $w_0$.
\end{proof}

\begin{prop}
$\uptau^{(1)}$ is $W(E_7)$-invariant.
\end{prop}

\begin{proof}

Let us start with the manifestly $\mathfrak{S}_8$-symmetric
expression for $\uptau^{(1)}$ previously given in 
\eqref{tau_1}.
Since $W(E_7)$ is generated by $w_0$ and $\mathfrak{S}_8$ 
we only need to check invariance under $w_0$,
as in the proof of Proposition \ref{prop.tau0_E7}.
The difference from the proof there is that now we have $X\in D_1\subset H_{\tau+\sigma}\times H_r$.

Since $Q(X)$ \eqref{Q_def} is defined from the $W(E_7)$-invariant
bilinear form, 
we only need to check the $W(E_7)$-invariance of the factors
\begin{align}
\label{Gamma_I}
I(x, \tilx; \sigma, \tau) \prod_{0\le i<j\le 7} \Gt_{\sigma\tau}( x_i+x_j, \tilx_i+\tilx_j ; \sigma, \tau, \tau,\mu)   .
\end{align}

The Weyl reflection  $w_0 \in W(E_7)$ acts on the coordinates $X\in D_1\subset H_{\tau+\sigma}\times H_r$ as
\begin{align}
w_0(x_i)=
\begin{cases}
x_i+\frac{\sigma+\tau}{2}-\frac{1}{2}\sum_{i=0}^3x_i & i \in \{0,1,2,3\}, \\
x_i+\frac{\sigma+\tau}{2}-\frac{1}{2}\sum_{i=4}^7x_i & i \in \{4,5,6,7\},
\end{cases}
\end{align}
\begin{align}
w_0(\tilx_i)=
\begin{cases}
\tilx_i+\frac{r}{2}-\frac{1}{2}\sum_{i=0}^3\tilx_i & i \in \{0,1,2,3\}, \\
\tilx_i+\frac{r}{2}-\frac{1}{2}\sum_{i=4}^7\tilx_i & i \in \{4,5,6,7\}.
\end{cases}
\end{align}
Note that this $w_0$ transformation takes the same form as the transformation rule of $I(x, \tilx; \sigma, \tau)$, given in \eqref{tilde}, 
where the coordinates of $X\in D_1\subset H_{\tau+\sigma}\times H_r$,
also exactly satisfy the balancing condition
\begin{align}
\sum_{i=0}^7 \vt_i=2(\sigma+\tau)\,,\qquad\sum_{i=0}^7 \smb{x}_i=2r.
\end{align}
Thus the action of $w_0$, followed by the transformation \eqref{trans3}, gives
\begin{align}
\label{trans32}
w_0\left(I(x, \tilx; \sigma, \tau)\right)=I(x, \tilx; \sigma, \tau)\hspace{-0.2cm}\prod_{\substack{0\le i<j \le 3 \\ \textrm{ or } 4\le i<j \le 7 }}\left( \Gamma(\vt_i+\vt_j, \smb{\vt}_i+\smb{\vt}_j;\sigma,\tau )\right)^{-1}\hspace{-0.2cm}.
\end{align}
We also have
\begin{align}
w_0(x_i+x_j)=
\begin{cases} 
\sigma+\tau-x_i-x_j & (\{ i,j,k,l \}=\{0,1,2,3\} \textrm{ or } \{4,5,6,7\} ), \\
x_i+x_j & (i \in \{0,1,2,3\}, j\in \{4,5,6,7\}) , 
\end{cases}
\end{align}
and 
\begin{align}
w_0(\tilx_i+ \tilx_j)=
\begin{cases} 
r-\tilx_i-\tilx_j & (\{ i,j,k,l \}=\{0,1,2,3\} \textrm{ or } \{4,5,6,7\} ) ,\\
\tilx_i+\tilx_j & (i \in \{0,1,2,3\}, j\in \{4,5,6,7\}) , 
\end{cases}
\end{align}
which for $i,j\in\{0,1,2,3\}$ or $i,j\in\{4,5,6,7\}$, results in
\begin{align}
\begin{split}
&w_0 \left( \Gt_{\sigma\tau}({\vt}_i+{\vt}_j,{\smb{x}}_i+{\smb{x}}_j;\sigma,\tau,\tau,\mu) \right)\\
&\qquad=\Gt_{\sigma\tau}(\sigma+\tau-(\vt_k+\vt_l),r-(\smb{x}_k+\smb{x}_l);\sigma,\tau,\tau,\mu) \\
&\qquad=\Gt_{\sigma\tau}(\vt_k+\vt_l,\smb{x}_k+\smb{x}_l;\sigma,\tau,\tau,\mu)\,\Gamma(\vt_k+\vt_l,\smb{x}_k+\smb{x}_l;\sigma,\tau),
\end{split}
\end{align}
where in the last line we have used \eqref{shift_Gamma3}.
The contributions of the type $\Gamma(\vt_k+\vt_l,\smb{x}_k+\smb{x}_l;\sigma,\tau)$ in the last line, exactly cancel the
contribution coming from the factors in the product of \eqref{trans32}, and thus \eqref{Gamma_I} is invariant under the Weyl reflection $w_0$.
\end{proof}

\begin{prop}\label{prop.E7}
$\uptau^{(n)}$ is  $W(E_7)$-invariant for $n=0,1, \dots$.
\end{prop}

\begin{proof}
As stated in the proof of Proposition \ref{rem.recursion},
$\uptau^{(n)}$ is defined recursively from $\uptau^{(0)}$ and $\uptau^{(1)}$,
which we have shown already to be $W(E_7)$-invariant.
This proves the $W(E_7)$-invariance of $\uptau^{(n)}$.
\end{proof}

\begin{cor}
The product $\mathcal{G}^{(n)}(X) I_n(X)$,
with $\mathcal{G}^{(n)}(X)$ given in \eqref{G_def} and $I_n$
given in \eqref{tau_n_integral}, is $W(E_7)$-invariant.
\end{cor}

\begin{proof}
This immediately follows from Proposition \eqref{prop.E7},
the expression of $\uptau^{(n)}$ given in \eqref{tau_n_product},
and the $W(E_7)$-invariance of $e^{(n)}(X)$.
\end{proof}

\subsection{Bilinear Identities}

Having proven the $W(E_7)$-invariance of the $\uptau^{(n)}$,
we will now prove that the Hirota equations are satisfied.
For this purpose we start with a few lemmas concerning
the function $g^{(n)}$.

\begin{lem}\label{lem.ggg}
For a $C_3$-frame $\{\pm a_0, \pm a_1, \pm a_2 \}$ of type ${\rm II}_1$, 
we have
\begin{align}
\begin{split}
& g^{(n-1)}(X- a_0 \vT) g^{(n+1)}(X+ a_0 \vT)
=\frac{[(a_0\pm a_2 |X)]}{[(a_1\pm a_2 |X)]}  g^{(n)}(X\pm a_1 \vT), \\
& (X\in D_n, \quad n=1,2, \ldots) .
\end{split}
\label{G_rec}
\end{align}
\end{lem}

\begin{proof}

Since each $\uptau^{(n)}$ has manifest $W(E_7)$-symmetry (Proposition \ref{prop.E7}),
it is sufficient to write down the Hirota equations
\eqref{Hirota} for a special example of $C_3$-frame of type $\II_1$.
Let us choose the $C_3$-frame to be as given in \eqref{II1_frame}.

We compute the ratio $g^{(n-1)}(X- a_0 \vT) g^{(n+1)}(X+ a_0 \vT)/g^{(n)}(X\pm a_1 \vT)$
for each factor $e^{(n)}, d^{(n)}, \mathscr{G}^{(n)}$, of the function $g^{(n)}$ \eqref{g_def}.
For $e^{(n)}$ we find from the definition \eqref{e_def}
\begin{align}
\begin{split}
\frac{e^{(n-1)}(X-a_0\vT)e^{(n+1)}(X+a_0\vT)}{e^{(n)}(X\pm a_1\vT)}
&=\EXP^{\frac{2\pi\ii}{r}(\sigma+1)}e^{-(Q(X+a_0\vT)-Q(X-a_0\vT))} \\
&=\EXP^{\frac{2\pi\ii}{r}(\sigma+1- 2 (a_0|x) +2(a_0|\tilx))} \\
&=\EXP^{\frac{2\pi\ii}{r}(\sigma+1-x_0-x_1-x_2-x_3+\tilx_0+\tilx_1+\tilx_2+\tilx_3)} .
\end{split}
\end{align}
For $d^{(n)}$ we compute from the definition \eqref{d_def}, after many cancellations,
\begin{align}
\begin{split}
&\frac{d^{(n-1)}(X-a_0\vT)d^{(n+1)}(X+a_0\vT)}{d^{(n)}(X\pm a_1\vT)}\\
&\qquad=\EXP^{\frac{2\pi\ii}{r}(\sigma+1-x_0-x_1+\smb{x}_0+\smb{x}_1)}\theta_\sigma(X_0\pm X_3)\theta_\sigma(X_1\pm X_2),
\end{split}
\end{align}
Let us next compute the ratio 
\begin{align}
&\frac{\mathscr{G}^{(n-1)}(X-a_0\vT) \mathscr{G}^{(n+1)}(X+a_0\vT)}{\mathscr{G}^{(n)}(X\pm a_1\vT)},
\end{align}
for $\mathscr{G}^{(n)}$.
In this computation
most of the gamma function factors in the definition of $\mathscr{G}^{(n)}$ in \eqref{G_def} cancel out; the exceptions are the 
cross terms involving $\Gt_\sigma$, for $0\le i<j\le 3$ (hence the result is independent of the value of $n$). 
For example, for the term with $i=0, j=2$,
the expression $X\pm a_0 \vT$ gives $(x_0+x_2\pm \tau, \tilx_0+\tilx_2\pm 1)$
while $X\pm a_1 \vT$ gives $(x_0+x_2, \tilx_0+\tilx_2)$ twice, so that we have
\begin{align}
\begin{split}
&\frac{\Gt_\sigma(x_0+x_2+2 \tau, \tilx_0+\tilx_2+2; \sigma, \tau, \tau) \Gt_\sigma(x_0+x_2, \tilx_0+\tilx_2; \sigma, \tau, \tau)}{\Gt_\sigma(x_0+x_2+\tau, \tilx_0+\tilx_2+1; \sigma, \tau, \tau)^2}  \\
&=\frac{\EXP^{\phi_e(x_0+x_2+ \tau, \tilx_0+\tilx_2+ 1; \sigma, \tau, \tau)}}{\EXP^{\phi_e(x_0+x_2, \tilx_0+\tilx_2; \sigma, \tau, \tau)}}\frac{\gamma_\sigma(x_0+x_2+ \tau, \tilx_0+\tilx_2+ 1; \sigma, \tau, \tau)}{\gamma_\sigma(x_0+x_2, \tilx_0+\tilx_2; \sigma, \tau, \tau)}  \\
&=\frac{\Gamma(x_0+x_2+\tau, \tilx_0+\tilx_2+ 1; \sigma, \tau)}{\Gamma(x_0+x_2, \tilx_0+\tilx_2; \sigma, \tau)} 
\\
&=\theta_\sigma(x_0+x_2, \tilx_0+\tilx_2) ,
\end{split}
\label{ratio_simplify}
\end{align}
where we used \eqref{shift_Gamma3}, \eqref{litgamrels}, \eqref{legf2} and then \eqref{Gamma_ratio_2}. 
By repeating this manipulation we obtain
\begin{align}
\begin{split}
&\frac{\mathscr{G}^{(n-1)}(X-a_0\vT) \mathscr{G}^{(n+1)}(X+a_0\vT)}{\mathscr{G}^{(n)}(X\pm a_1\vT)}\\
&\qquad =
\theta_\sigma(X_0+X_2) \theta_\sigma(X_0+X_3 )\theta_\sigma(X_1+X_2 ) \theta_\sigma(X_1+X_3 ) .
\end{split}
\end{align}
Finally by combining  all of the above, we obtain
\begin{align}
\begin{split}
\frac{g^{(n-1)}(X-a_0\vT) g^{(n+1)}(X+a_0\vT)}{g^{(n)}(X\pm a_1\vT)}
&=
\EXP^{\frac{2\pi\ii}{r}(x_2+x_3-\tilx_2-\tilx_3)}
\frac{
\theta_\sigma(X_0+X_2) \theta_\sigma(X_1+X_3)}
{\theta_\sigma(X_0- X_3)\theta_\sigma(X_1- X_2)} \\
&=\frac{[(a_0\pm a_2 |X)]}{[(a_1\pm a_2 |X)]}.
\end{split}
\end{align}
\end{proof}

\begin{lem}\label{lem.gg}
For a $C_3$-frame $\{\pm a_0, \pm a_1, \pm a_2 \}$ of type ${\rm II}_1$, 
we have 
\begin{align}
\begin{split}
\frac{g^{(n)}(X\pm a_1 \vT)}{g^{(n)}(X\pm a_2 \vT)}
=\frac{[(a_0\pm a_1 |X)]}{[(a_0\pm a_2 |X)]} , \quad (X\in D_n),\; n=0,1,2,\dots .
\end{split}
\label{g_ratio}
\end{align}
\end{lem}

\begin{proof}

Let us first consider the case $n=0$.
Let us choose the $C_3$-frame of type $\II_1$ to be \eqref{II1_frame} as before.
In the previous computation of the ratio $g^{(0)}(X\pm a_1 \vT)/g^{(0)}(X\pm a_2 \vT)$
most of the gamma function factors cancels out; the exceptions are the 
cross terms involving $\gt_\sigma$, for $0\le i<j\le 3$. For example, for the term with $i=0, j=1$,
the expression $X\pm a_1 \vT$ gives $(x_0+x_1\pm \tau, \tilx_0+\tilx_1\pm 1)$
while $X\pm a_2 \vT$ gives $(x_0+x_1, \tilx_0+\tilx_1 )$ twice.
We can then appeal to the manipulations \eqref{ratio_simplify},
to obtain
\begin{align}
\begin{split}
\frac{g^{(0)}(X\pm a_1 \vT)}{g^{(0)}(X\pm a_2 \vT)}
&=\frac{\theta_\sigma(x_0+x_1, \tilx_0+\tilx_1; \tau)\, \theta_\sigma(x_2+x_3, \tilx_2+\tilx_3; \tau)}{\theta_\sigma(x_0+x_2 , \tilx_0+\tilx_2; \tau)\, \theta_\sigma(x_1+x_3, \tilx_1+\tilx_3; \tau)} \\
&=\frac{[X_0+X_1][X_2+X_3]}{[X_0+X_2][X_1+X_3]}=\frac{[(a_0\pm a_1 |X)]}{[(a_0\pm a_2 |X)]} .
\end{split}
\end{align}

The case of $n=1$ is similar. We can use the expression for $g^{(1)}$ coming from \eqref{g_def},
\begin{align}
\begin{split}
&g^{(1)}(X) = e^{-Q(X)}
 \prod_{\substack{{0\le i \le 3 } \\ { 4\le j \le 7}}} \Gt_{\sigma\tau}(x_i+x_j, \tilx_i+\tilx_j ; \sigma, \tau, \tau,\mu) \\
& \times \hspace{-0.1cm} \prod_{\substack{0\le i<j \le 3 \\ \textrm{ or } 4\le i<j \le 7}} \hspace{-0.1cm} \Gt_\sigma( \tau+x_i+x_j,1+ \tilx_i+\tilx_j ; \sigma, \tau, \tau)\, \gt_\tau(x_i+x_j+\mu;,\tilx_i+\tilx_j;\sigma,\tau,\mu) .
\end{split}
\end{align}
In the computation of the ratio $g^{(0)}(X\pm a_1 \vT)/g^{(0)}(X\pm a_2 \vT)$ the contribution from the prefactor $e^{-Q(X)}$
cancels out, and the only relevant part
which remains after taking the ratio is 
\begin{align}
 \prod_{0\le i<j \le 3 } \Gt_{\sigma}(\tau+ x_i+x_j, 1+ \tilx_i+\tilx_j ; \sigma, \tau, \tau).
\end{align}
This is exactly the same factor from the definition of $g^{(0)}$ that contributes to the case of $n=0$,
and hence the computation will goes through exactly the same as for the $n=0$ case, to obtain
\eqref{g_ratio} for $n=1$.

For the case $n>1$ we use the induction with respect to the integer $n$.
From \eqref{G_rec} we have
\begin{align}
\begin{split}
g^{(n+1)}(X+ a_0 \vT)
&=\frac{g^{(n)}(X\pm a_1 \vT)}{g^{(n-1)}(X- a_0 \vT) } \frac{[(a_0\pm a_2 |X)]}{[(a_1\pm a_2 |X)]}\\
& =\frac{g^{(n)}(X\pm a_2 \vT)}{g^{(n-1)}(X- a_0 \vT) } \frac{[(a_0\pm a_1 |X)]}{[(a_1\pm a_2 |X)]} ,
\end{split}
\end{align}
where in the last line we used the assumption for $n$, to obtain yet another expression for 
$g^{(n+1)}$. Shifting the value of $X$, and using one of the two expressions above, gives
\begin{align}
\begin{split}
g^{(n+1)}(X\pm a_1 \vT)
&=\frac{g^{(n)}(X-a_0 \vT \pm a_1 \vT\pm a_2 \vT)}{g^{(n-1)}(X- 2 a_0 \vT\pm a_1 \vT) } \frac{[(a_0\pm a_1 |X)- \vT\pm \vT]}{[(a_1\pm a_2 |X)\pm  \vT]} , \\
g^{(n+1)}(X\pm a_2 \vT)
&=\frac{g^{(n)}(X-a_0 \vT\pm a_1 \vT\pm a_2 \vT)}{g^{(n-1)}(X- 2a_0 \vT\pm a_2 \vT) } \frac{[(a_0\pm a_2 |X)- \vT\pm  \vT]}{[(a_1\pm a_2 |X) \pm  \vT]},
\end{split}
\end{align}
where we used $(a_i|a_j)=\delta_{ij}$. We thus obtain 
\begin{align}
\begin{split}
\frac{g^{(n+1)}(X\pm a_1 \vT)}{g^{(n+1)}(X\pm a_2 \vT)} 
& =\frac{g^{(n-1)}(X- 2a_0 \vT\pm a_2 \vT)}{g^{(n-1)}(X- 2 a_0 \vT\pm a_1 \vT)}
\frac{[(a_0\pm a_1 |X)- \vT\pm \vT]}{[(a_0\pm a_2 |X)- \vT\pm  \vT]} \\
& =\frac{g^{(n-1)}(X- 2a_0 \vT\pm a_2 \vT)}{g^{(n-1)}(X- 2 a_0 \vT\pm a_1 \vT)}
\frac{[(a_0\pm a_1 |X)- 2 \vT]}{[(a_0\pm a_2 |X)- 2 \vT]} 
 \frac{[(a_0\pm a_1 |X)]}{[(a_0\pm a_2 |X)]} \\
& =  \frac{[(a_0\pm a_1 |X)]}{[(a_0\pm a_2 |X)]} ,
\end{split}
\end{align}
where in the last line we used the assumption for $n-1$.  This is what we wanted to show.
\end{proof}

We finally come to the proof that the $\uptau^{(n)}$-functions satisfy the desired bilinear identities. 
We first prove the special cases $({\rm II}_2)_0$, $({\rm II}_1)_0$, $({\rm II}_0)_0$, $(\I)_{1/2}$.
We then prove $(\II_1)_{n=1,2, \dots}$, from which the remaining cases will follow, by Proposition \ref{prop.follows}.

\begin{prop}
$({\rm II}_2)_0$ holds. 
\end{prop}

\begin{proof}
The identity $({\rm II}_2)_0$ reads
\begin{align}
[(a_0 \pm a_1| X)] \tau^{(0)}(X\pm a_2)=0,
\end{align}
for a $C_3$-frame $\{\pm a_0, \pm a_1, \pm a_2 \}$
with $(\phi|a_0)=(\phi|a_1)=1, (\phi|a_2)=0$ (recall \eqref{phi_a}). 
Since $a_0+a_1=\phi$ (see Remark \ref{rem.enlarge_C3}),
one finds for $X\in D_0$ that
\begin{align}
[(a_0 \pm a_1| X)] =[(\phi|X)]=[(\sigma, -1)]=e^{\phi_{\sigma}(\sigma, -1)} \theta(1|\, e^{2\pi \ii \sigma r}) =0 .
\end{align}
\end{proof}

\begin{prop}
$({\rm II}_1)_0$ holds. 
\end{prop}

\begin{proof}
This is an immediate consequence of \eqref{g_ratio}, since $\tau^{(0)}=g^{(0)}$.
\end{proof}

\begin{prop}
$(\II_0)_0$ holds.
\end{prop}

\begin{proof}
For a $C_3$-frame $\{\pm a_0,\pm a_1,\pm a_2 \}$  of type $(\II_0)_0$, we can 
choose an extra vector $a_3$ such that $\{\pm a_i,\pm a_j,\pm a_3 \}$
for $0\le i<j\le 2$, are all $C_3$-frames of type  $(\II_1)$ (see Remark \ref{rem.enlarge_C3}).
This implies that by using \eqref{g_ratio}, for $\tau^{(0)}=g^{(0)}$ 
\begin{align}
\tau^{(0)}(X \pm a_i T)=\frac{[(a_3\pm a_i | X)]}{ [(a_3\pm a_j | X)]} \tau^{(0)}(X \pm a_j T) \quad (0\le i<j\le 2) .
\end{align}
The identity in question, namely $(\II_0)_0$ (\eqref{II_0_n} with $n=0$), now reduces to
\begin{align}
\begin{split}
 &[(a_1\pm a_2 | X)]  [(a_3\pm a_0 | X)]  +[(a_2\pm a_0 | X)] [(a_3\pm a_1 | X)] \\
& \qquad +[(a_0\pm a_1 | X)] [(a_3\pm a_2 | X)]=0 .
\end{split}
\end{align}
This holds due to the three-term identity \eqref{three-term}.
\end{proof}

\begin{prop}
$(\I)_{1/2}$ holds.
\end{prop}

\begin{proof}

As a $C_3$-frame of type $(\I)_{1/2}$ we can choose from the $W(E_7)$-orbit
a representative $\{ \pm a_0, \pm a_1, \pm a_2 \}=\{ \pm v_1, \pm v_2, \pm v_3 \}$.
We wish to show
\begin{align}
\begin{split}
&[X_2\pm X_3] \uptau^{(0)}(X-v_0 \vT) \uptau^{(1)}(X+v_0 \vT)
+[X_3\pm X_1] \uptau^{(0)}(X-v_0 \vT)\uptau^{(1)}(X+v_1 \vT) \\
&+[X_1\pm X_2] \uptau^{(0)}(X-v_2 \vT) \uptau^{(1)}(X+v_2 \vT) =0.
\end{split}
\label{show1}
\end{align}
Let us define 
\begin{align}
\begin{split}
\mathcal{F}(X)&:= \prod_{0\le i<j\le 7}  \Gt_\sigma( x_i+x_j , \tilx_i+\tilx_j ; \sigma, \tau, \tau),\\
J(X)&:=e^{-Q(X)}  I(X)\prod_{0\le i<j\le 7}  \gt_\tau( x_i+x_j , \tilx_i+\tilx_j ; \sigma, \tau, \mu) ,
\end{split}
\end{align}
so that 
\begin{align}
\begin{split}
\uptau^{(0)}&=\mathcal{F}(X+\vT),\qquad \uptau^{(1)}=\mathcal{F}(X) J(X) .
\end{split}
\end{align}
We see that the following two ratios
\begin{align}
\begin{split}
&\frac{\mathcal{F}(X+\vT-v_0 \vT) }{\mathcal{F}(X+\vT)}= \prod_{0<j\le 7} 
\frac{ \Gt_\sigma( x_0+x_j , \tilx_0+\tilx_j ; \sigma, \tau, \tau)}{ \Gt_\sigma(\tau+ x_0+x_j , 1+\tilx_0+\tilx_j  ; \sigma, \tau, \tau)}  ,  \\
&\frac{\mathcal{F}(X+v_0 \vT) }{\mathcal{F}(X)}= \prod_{0< j\le 7} 
\frac{ \Gt_\sigma(\tau+ x_0+x_j , 1+\tilx_0+\tilx_j  ; \sigma, \tau, \tau)}{ \Gt_\sigma( x_0+x_j , \tilx_0+\tilx_j ; \sigma, \tau, \tau)}  , 
\end{split}
\end{align}
are the inverse of each other. This means that
\begin{align}
\begin{split}
\uptau^{(0)}(X-v_0 \vT) \uptau^{(1)}(X+v_0 \vT) &=
\mathcal{F}(X+\vT-v_0 \vT)  \mathcal{F}(X+v_0 \vT) J(X+v_0 \vT)
\\
&=\mathcal{F}(X+\vT)  \mathcal{F}(X) e^{-Q(X+v_0 \vT)} I(X+v_0 \vT),
\end{split}
\label{tautau}
\end{align}
Substituting \eqref{tautau} into \eqref{show1},
\eqref{show1}, gives
\begin{align}
\begin{split}
&[X_1\pm X_2] e^{-Q(X+v_0 \vT)} I(X+v_0 \uptau)
+[X_2\pm X_0]  e^{-Q(X+v_1 \vT)} I(X+v_1 \vT) \\
&+[X_0\pm X_1] e^{-Q(X+v_2 \vT)} I(X+v_2 \vT) =0.
\end{split}
\label{show2}
\end{align}
Next from the definition of $Q(X)$ in \eqref{Q_def}, we obtain
\begin{align}
Q(X+v_i \vT)=Q(X)+\frac{2\pi \ii}{r} \left( x_i-\tilx_i+\frac{\tau-1}{2} \right),
\end{align}
and consequently \eqref{show2} reduces to
\begin{align}
\label{contrel2}
\begin{split}
&\left[X_1\pm X_2\right]\EXP^{-\frac{2\pi\ii}{r}(\vt_i-\smb{x}_i)}T_{\tau,i}I(t,a) \\
&\qquad +\left[X_2\pm X_0\right]\EXP^{-\frac{2\pi\ii}{r}(\vt_j-\smb{x}_j)}T_{\tau,j}I(t,a)  \\
&\qquad +\left[X_0\pm X_1\right]\EXP^{-\frac{2\pi\ii}{r}(\vt_k-\smb{x}_k)}T_{\tau,k}I(t,a)=0,
\end{split}
\end{align}
where the shift operator $T_{\tau,k}$ is defined in \eqref{T_shift}.
This is exactly the  contiguity relation \eqref{contrel}.
\end{proof}

\begin{prop}
$(\II_1)_n$ holds for $n=1,2,\dots$.
\end{prop}

\begin{proof}
For a $C_3$-frame $\{\pm a_0, \pm a_1, \pm a_2 \}$ of type ${\rm II}_1$, we obtain
from Lemma \ref{lem.ggg} and Lemma \ref{lem.gg} that for $n=0, 1, \dots$,
\begin{align}
\begin{split}
& [(a_1\pm a_2 | X)] g^{(n-1)} (X- a_0\vT) g^{(n+1)} (X+ a_0\vT) \\
& =[(a_2\pm a_0 | X)]  g^{(n)} (X\pm a_1\vT) \\
& =[(a_0\pm a_1 | X)] g^{(n)} (X\pm  a_2\vT) .
\end{split}
\end{align}
This means that each of the above factors may be cancelled out of $({\rm II}_1)_{n}$,
resulting in
\begin{align}
\begin{split}
& K^{(n-1)} (X- a_0\vT) K^{(n+1)} (X+ a_0\vT) \\
& + K^{(n)} (X\pm a_1\vT) 
+ K^{(n)} (X\pm  a_2\vT) =0 , \quad (n=1,2,\ldots) ,
\end{split}
\end{align}
where $K^{(n)}(x):=\det\left(\psi^{(n)}_{ij}(x)\right)^n_{i,j=1}$ is the Casorati determinant.
The last equation is satisfied as a consequence of the Lewis Carroll formula.
\end{proof}

\appendix

\section{\texorpdfstring{Derivation of $W(E_7)$ Sum/Integral Transformation}{Derivation of W(E7) Sum/Integral Transformation}}\label{app.e7trans}

In this Appendix a proof will be given of Proposition \ref{prop.e7trans}.  Note that in the following, the function $\Gamma(z,\tilz)$ denotes the lens elliptic gamma function defined in \eqref{legf2}.  First consider the following $A_1 \leftrightarrow A_0$ transformation of \cite{Kels:2017toi}.  

\begin{prop}
Suppose that  $t=(t_0,\ldots,t_5)\in \mathbb{C}^6$, $\im(t_i)>0$, and $\tilt=(\tilt_0,\ldots,\tilt_5)\in \mathbb{Z}^6$, satisfy
\begin{align}
\label{balancing}
\sum_{i=0}^{5} t_i\equiv\sigma+\tau\quad (\bmod \, 2r)\,,\quad\sum_{i=0}^{5} \tilt_i \equiv0\quad (\bmod \, r)\,.
\end{align}
Then the sum/integral
\begin{align}
\label{ebsi}
\begin{split}
I_0(\sumY|t,\tilt)
=\frac{\lambda}{2} \hspace{-0.2cm} \sum_{\substack{\tilz_0,\tilz_1=0 \\[0.1cm] \sum_{i=0}^{1}\tilz_i=\sumY}}^{r-1} \hspace{-0.6cm}\int\limits_{\substack{ \hspace{0.4cm}[0,1]^2 \\[0.1cm] \hspace{0.8cm}\sum_{i=0}^1z_i=0}}\hspace{-0.8cm}dz_0\,dz_1 
\frac{\prod\limits_{i=0}^1\prod\limits_{j=0}^{2}\Gamma(t_j+z_i, \tilt_j+\tilz_i) \prod\limits_{j=3}^{5}\Gamma(t_j-z_i, \tilt_j-\tilz_i)}{\Gamma\left(\pm(z_0-z_1),\pm(\tilz_0-\tilz_1)\right)},
\end{split}
\end{align}
where $\sumY\in\mathbb{Z}$, and $\lambda$ is defined in \eqref{lambda_def}, may be evaluated as
\begin{align}
\label{eq2}
\begin{split}
I_0(\sumY|t,\tilt)=\prod_{\substack{0\leq i \leq 2 \\ 3\leq j\leq 5}}\Gamma(t_i+t_j,\tilt_i+\tilt_j)\prod_{0\leq i<j\leq 2}\Gamma(t_i+t_j,\tilt_i+\tilt_j+\sumY)\phantom{.}\\
\times 
\prod_{3\leq i<j\leq 5}\Gamma(t_i+t_j,\tilt_i+\tilt_j-\sumY) .
\end{split}
\end{align}
\end{prop}
Note that $\tilz_0,\tilz_{1}$, in the sum \eqref{ebsi}, are regarded as elements of $\mathbb{Z}/r \mathbb{Z}$.
For $\sumY=0$, \eqref{eq2} is the elliptic beta sum/integral formula \cite{Kels:2015bda}.

We wish to extend the above formula to the case of
$\tilt\in \mathbb{Z}^8 \cup (\mathbb{Z}+1/2)^8$.  
For this purpose, we set
\begin{align}\label{ZY}
\sumY=\begin{cases}   
0 & (\tilt\in\mathbb{Z}^8),  \\
r & \left(\tilt\in \left(\mathbb{Z}+\frac{1}{2}\right)^8 \right). \\
\end{cases}
\end{align}
The sum over $\tilz_0, \tilz_1$, satisfying $\tilz_0+\tilz_1=\sumY$ can be exchanged for a
sum over a new variable $\tilz$, where
\begin{align}
\begin{cases}
\tilz_0=\tilz,\quad \tilz_1= -\tilz  , &\quad  (\tilt\in\mathbb{Z}^8, \sumY=0) , \\
\tilz_0=\tilz+\frac{r}{2},\quad \tilz_1= -\tilz+\frac{r}{2}  , &\quad  (\tilt\in \left(\mathbb{Z}+\frac{1}{2}\right)^8, \sumY=r;\, r \textrm{ even}) , \\
\tilz_0=\tilz +\frac{r+1}{2},\quad \tilz_1= -\tilz+\frac{r-1}{2}  , &\quad  (\tilt\in \left(\mathbb{Z}+\frac{1}{2}\right)^8, \sumY=r; \, r \textrm{ odd}) .
\end{cases}
\end{align}
This may be concisely written as
\begin{align}
\label{Zchoice}
\begin{split}
\tilz_0 &=+\tilz+(r+((r+1)\bmod 2))(\tilt_1 \bmod 1),\quad \\
\tilz_1 &=-\tilz+(r-((r+1)\bmod 2))(\tilt_1 \bmod 1).
\end{split}
\end{align}

Using the $r$-periodicity of the lens elliptic gamma function, with the choice \eqref{ZY}, \eqref{Zchoice}, the sum/integral \eqref{ebsi} may be written as
\begin{align}
I_0(t,\tilt):=I_0(\sumY|t,\tilt)=\frac{\lambda}{2}\sum_{\substack{\tilz=0}}^{r-1}\int_{[0,1]}dz
&\frac{\prod_{j=0}^{5}\Gamma(t_j\pm z,\tilt_j\pm \tilz_0)}{\Gamma(\pm 2z,\pm 2\tilz_0)},
\end{align}
while the formula \eqref{eq2} becomes
\begin{align}
\label{eq22}
I_0(t,\tilt)&=\prod_{0\leq i\leq 5}\Gamma(t_i+t_j,\tilt_i+\tilt_j).
\end{align}
Note that \eqref{eq22} is valid for both $\tilt\in\mathbb{Z}^8$ and  $\tilt\in (\mathbb{Z}+\frac{1}{2})^8$. 

Proposition \ref{prop.e7trans} may be proven with the use of \eqref{eq22}, analogously to a derivation given by Spiridonov \cite{rarified}.\footnote{The derivation appearing in \cite{rarified} uses a different notation than is used here, but after a change of variables, both derivations are seen to be based on the same elliptic beta sum/integral formula \eqref{eq2} that was first proven by the authors \cite{Kels:2015bda,Kels:2017toi}.  Thus the derivations are equivalent.  The authors thank V.P. Spiridonov for pointing this out.}

\begin{proof}[Proof of Proposition \ref{prop.e7trans}]

Consider $\alpha\in\mathbb{C}$, $x,y\in\mathbb{C}^4$, and $\tilx,\tily\in\mathbb{Z}^4\cup (\mathbb{Z}+\frac{1}{2})^4$, where $\im(\alpha),\im(x_i),\im(y_i)>0$, and
\begin{align}
2\alpha+\sum_{j=0}^3 x_j=2\alpha+\sum_{j=0}^3 y_j=\sigma+\tau,\quad 2\bbalpha+\sum_{j=0}^3 \tilx_j=2\bbalpha+\sum_{j=0}^3 \tily_j  = k r,
\end{align}
for some integer $k$.  If both $\tilx,\tily\in\mathbb{Z}^4$ or both $\tilx,\tily\in(\mathbb{Z}+\frac{1}{2})^4$, then we choose $\bbalpha\in\mathbb{Z}$, otherwise we choose $\bbalpha\in(\mathbb{Z}+\frac{1}{2})$.

In terms of the above variables, consider the following sum/integral
\begin{align}
\label{eq32}
\begin{split}
&\sum_{\substack{\tilz=0}}^{r-1}\sum_{\substack{\tilw=0}}^{r-1}\int_{[0,1]^2}\hspace{-0.2cm} dwdz\,\Gamma(\alpha\pm z\pm w,\bbalpha\pm \hat{\tilz}\pm \hat{\tilw})\\[-0.2cm]
& \qquad\qquad \qquad \times \frac{\prod_{j=0}^3\Gamma(x_j\pm z, \tilx_j\pm \hat{\tilz})\,\Gamma(y_j\pm w,\tily_j\pm \hat{\tilw})}{\Gamma(\pm 2z,\pm 2 \hat{\tilz})\,\Gamma(\pm 2w,\pm 2 \hat{\tilw})},
\end{split}
\end{align}
where
\begin{align}
\begin{split}
\hat{\tilz} &=+\tilz+(r+((r+1)\bmod 2))(\tilx_1 \bmod 1),\quad \\
\hat{\tilw} &=+\tilw+(r+((r+1)\bmod 2))(\tily_1 \bmod 1).
\end{split}
\end{align}

The expression \eqref{eq32} may be sum/integrated in two different ways.  First using \eqref{eq22} to sum/integrate \eqref{eq32} over $z, \tilz$, gives
\begin{align}
\label{eq52}
\begin{split}
&\Gamma(2\alpha,2\bbalpha)\prod_{0\leq i<j\leq 3}\Gamma(x_i+x_j, \tilx_i+\tilx_j)\\
&\qquad \times \sum_{\substack{\tilw=0}}^{r-1}\int_{[0,1]}dw \frac{\prod_{j=0}^3\Gamma(\alpha\pm w+x_j,\bbalpha\pm \hat{\tilw}+\tilx_j)\,\Gamma(y_j\pm w,\tily_j\pm \hat{\tilw})}{\Gamma(\pm 2w,\pm 2 \hat{\tilw})}.
\end{split}
\end{align}
Next using \eqref{eq22} to sum/integrate \eqref{eq32} over $w, \tilw$, gives
\begin{align}
\label{eq62}
\begin{split}
&\Gamma(2\alpha,2\bbalpha)\prod_{0\leq i<j\leq 3}\Gamma(y_i+y_j, \tily_i+\tily_j) \\
&\qquad\times\sum_{\substack{\tilz=0}}^{r-1}\int_{[0,1]}dz\frac{\prod_{j=0}^3 \Gamma(\alpha\pm z+y_j,\bbalpha\pm \hat{\tilz}+\tily_j)\,\Gamma(x_j\pm z,\tilx_j\pm \hat{\tilz})}{\Gamma(\pm 2z,\pm 2 \hat{\tilz})}.
\end{split}
\end{align}

Define the variables $t=(t_0,\ldots,t_7)\in\mathbb{C}^8$, and $\tilt=(\tilt_0,\ldots,\tilt_7)\in\mathbb{Z}^8\cup (\mathbb{Z}+\frac{1}{2})^8$, as
\begin{align}
\label{vars2}
\begin{split}
&t_i=\alpha+x_i,  \quad  t_{i+4}=y_i  , \quad\tilt_i=\bbalpha+\tilx_i,   \quad \tilt_{i+4}=\tily_i , \quad  (i=0, \dots, 3),
\end{split}
\end{align}
and $I(x,\tilx)$ as the following sum/integral
\begin{align}
\label{sumint}
I(x,\tilx) =
\frac{\lambda}{2}\sum_{\substack{\tilz=0}}^{r-1}\int_{[0,1]} dz\frac{\prod_{j=0}^{7}\Gamma(x_j\pm z,\tilx_j\pm \hat{\tilz})}{\Gamma(\pm 2z,\pm 2 \hat{\tilz})}.
\end{align}
With these variables, the sum/integral appearing in \eqref{eq52} is given by
\begin{align}
I(t,\tilt),
\end{align}
while the sum/integral appearing in \eqref{eq62} is given by
\begin{align}
I(t',\tilt'),
\end{align}
where 
\begin{align}
\label{e7intt2}
\begin{split}
t'_i&=t_i-\alpha=t_i+\frac{\sigma+\tau}{2}-\frac{1}{2}\sum_{i=0}^3 t_i,\quad i=0,1,2,3, \\
t'_i&=t_i+\alpha=t_i+\frac{\sigma+\tau}{2}-\frac{1}{2}\sum_{i=4}^7 t_i,\quad i=4,5,6,7, \\
\tilt'_i&=\tilt_i-\bbalpha=\tilt_i+\frac{kr}{2}-\frac{1}{2}\sum_{i=0}^3\tilt_i,\quad i=0,1,2,3, \\
\tilt'_i&=\tilt_i+\bbalpha=\tilt_i+\frac{kr}{2}-\frac{1}{2}\sum_{i=4}^7\tilt_i,\quad i=4,5,6,7.
\end{split}
\end{align}
For the factors of the lens elliptic gamma functions that appear outside the sum/integrals in \eqref{eq52}, and \eqref{eq62}, we have for distinct $i,j\in\{0,1,2,3\}$,
\begin{align}
x_i+x_j=t_i+t_j-2\alpha,\quad & \tilx_i+\tilx_j=\tilt_i+\tilt_j-2\bbalpha,
\end{align}
leading to, for $i,j,k,l=\{0,1,2,3\}$,
\begin{align}
x_k+x_l=\sigma+\tau-t_i-t_j,\quad \tilx_k+\tilx_l=kr-\tilt_i-\tilt_j.
\end{align}
Equating \eqref{eq52}, with \eqref{eq62}, and collecting all factors, 
we finally obtain
\begin{align}
\label{e7int2}
I(t,\tilt)=I(t',\tilt')\prod_{0\leq i<j\leq 3}\Gamma(t_i+t_j,\tilt_i+\tilt_j)\prod_{4\leq i<j\leq 7}\Gamma(t_i+t_j,\tilt_i+\tilt_j),
\end{align}
where $t'$, and $\tilt'$ are given in \eqref{e7intt2}.  
\end{proof}

\section{Multiple Bernoulli Polynomials}\label{app.Bernoulli}

Let us define the multiple Bernoulli polynomials $B_{n,k}(z;\omega_1,\ldots,\omega_n)$ via the generating function
\begin{align}
\label{berngen}
\frac{x^n\,\EXP^{zx}}{\prod_{j=1}^n(\EXP^{\omega_j x}-1)}=\sum_{k=0}^\infty B_{n,k}(z;\omega_1,\ldots,\omega_n)\,\frac{x^k}{k!}\,,
\end{align}
where $z\in\mathbb{C}$, and $\omega_1,\ldots,\omega_n\in\mathbb{C}-\{0\}$.
These functions previously appeared in relation to the modular properties of multiple gamma functions \cite{Narukawa2004247}.  For this paper, only
 two particular multiple Bernoulli polynomials are needed.  
 
One of these is $B_{3,3}(z;\omega_1,\omega_2,\omega_3)$, which is given explicitly by
\begin{align}
\label{bernoulli}
\begin{split}
\ds B_{3,3}(z;\omega_1,\omega_2,\omega_3)&=\frac{z^3}{\omega_1\omega_2\omega_3}-\frac{3z^2\sum\limits_{i=1}^3\omega_i}{2\omega_1\omega_2\omega_3}+\frac{z\left(\sum\limits_{i=1}^3\omega_i^2+3\sum\limits_{1\leq i<j\leq 3}\omega_i\omega_j\right)}{2\omega_1\omega_2\omega_3}
\\
&-\frac{\left(\sum\limits_{i=1}^3\omega_i\right)\left(\sum\limits_{1\leq i<j\leq3}\omega_i\omega_j\right)}{4\omega_1\omega_2\omega_3}\,.
\end{split}
\end{align}
The other is $B_{4,4}(z;\omega_1,\omega_2,\omega_3,\omega_4)$, which is given by
\begin{align}
\label{bernoulli4}
\begin{split}
& B_{4,4}(z;\omega_1,\omega_2,\omega_3, \omega_4)\\
&= \frac{z^4}{\prod\limits_{i=1}^4\omega_i}-\frac{2z^3\sum\limits_{i=1}^4\omega_i}{\prod_{i=1}^4\omega_i}+\frac{z^2\left(\sum\limits_{i=1}^4\omega_i^2+3\sum\limits_{1\leq i<j\leq 4}\omega_i\omega_j\right)}{\prod\limits_{i=1}^4\omega_i}
-\frac{z\left(\sum\limits_{i=1}^4\omega_i\right)\left(\sum\limits_{1\leq i<j\leq3}\omega_i\omega_j\right)}{\prod\limits_{i=1}^4\omega_i} \\
&\qquad-\frac{\sum\limits_{i=1}^4\omega_i^4-5\sum\limits_{1\leq i<j\leq4}(\omega_i\omega_j)^2-15\sum\limits_{i=1}^4\sum\limits_{\substack{1\leq j<k\leq 4 \\ j,k\neq i}}\omega_i^2\omega_j\omega_k-45\prod\limits_{i=1}^4\omega_i}{30\prod\limits_{i=1}^4\omega_i}\,.
\end{split}
\end{align}
The above two multiple Bernoulli polynomials are related by
\begin{align}
B_{4,4}(z+\omega_4;\omega_1,\omega_2,\omega_3,\omega_4) - B_{4,4}(z;\omega_1,\omega_2,\omega_3,\omega_4)
=4 B_{3,3}(z;\omega_1,\omega_2,\omega_3) .
\end{align}

\bibliography{painleve}
\bibliographystyle{utphys}

\end{document}